\documentclass[pdflatex,sn-mathphys-num]{sn-jnl}

\usepackage{amsmath}
\usepackage{amssymb}
\usepackage{comment}

\usepackage[polutonikogreek,english]{babel}

\usepackage{amsfonts}
\usepackage{amssymb}
\usepackage{mathtools}

\usepackage{color}

\usepackage{hyperref}

\renewcommand{\theequation}{\arabic{section}.\arabic{equation}}

\theoremstyle{plain}
\newtheorem{theorem}{Theorem}[section]
\newtheorem{Proposition}[theorem]{Proposition}

\theoremstyle{remark}

\newcommand{\op}{\overline{\partial}}

\newcommand{\be}{\begin{equation}}
\newcommand{\bq}{\begin{equation}}
\newcommand{\eq}{\end{equation}}
\newcommand{\ee}{\end{equation}}
\newcommand{\ba}{\begin{eqnarray}}
\newcommand{\ea}{\end{eqnarray}}
\newcommand{\ce}{\coloneqq}
\newcommand{\lp}{\left(}
\newcommand{\rp}{\right)}
\newcommand{\lb}{\left[}
\newcommand{\rb}{\right]}
\newcommand{\la}{\left\{}
\newcommand{\ra}{\right\}}

\newcommand{\dd}{\mathrm{d}}
\def\PII{\mathrm{P}_{\mathrm{II}}}

\newcommand{\PVI}{\mathrm{P}_{\mathrm{VI}}}
\newcommand{\CVI}{\mathrm{C}_{\mathrm{VI}}}
\newcommand{\HVI}{H_{\mathrm{VI}}}

\newcommand{\EVI}{\mathrm{E}_{\mathrm{VI}}}

\newcommand{\Ksech}{K_{\mathrm{sech}}}

\newcommand{\Kta}{K_{\theta}}

\newcommand{\ii}{\mathrm{i}}
\newcommand{\bra}[1]{\langle#1\rvert} 
\newcommand{\ket}[1]{\lvert#1\rangle} 

\newcommand{\abcd}{\lb \alpha,\beta,\gamma,\delta \rb}
\newcommand{\upabcd}{\la \va,\vb,\vc,\vd \ra}

\newcommand{\lvabcd}{\la \vartheta_\alpha,\vartheta_\beta,\vartheta_\gamma,\vartheta_\delta \ra}

\newcommand{\ta}{\theta}

\newcommand{\upta}{\vartheta}
\newcommand{\taphi}{{\theta, \varphi}}

\newcommand{\upb}{\vb}
\newcommand{\upc}{\vc}

\newcommand{\vta}{\vartheta}
\newcommand{\va}{\vartheta_\alpha}
\newcommand{\vb}{\vartheta_\beta}
\newcommand{\vc}{\vartheta_\gamma}
\newcommand{\vd}{\vartheta_\delta}

\newcommand{\vabcd}{\la \va, \vb, \vc, \vd \ra}
\newcommand{\vava}{{\vartheta^2}}
\newcommand{\gta}{g_{\vava}}
\newcommand{\hta}{h_{\vava}}
\newcommand{\vtaxi}{\varphi_{\theta}(\xi)}
\newcommand{\fs}[1]{\frac{1}{#1}}

\newcommand{\zb}{\overline{z}}

\newcommand{\calA}{{\mathcal{A}}}
\newcommand{\calB}{{\mathcal{B}}}
\newcommand{\calC}{{\mathcal{C}}}
\newcommand{\calD}{{\mathcal{D}}}
\newcommand{\calE}{{\mathcal{E}}}
\newcommand{\calF}{{\mathcal{F}}}
\newcommand{\calG}{{\mathcal{G}}}
\newcommand{\calH}{{\mathcal{H}}}
\newcommand{\calJ}{{\mathcal{J}}}
\newcommand{\calK}{{\mathcal{K}}}
\newcommand{\calL}{{\mathcal{L}}}
\newcommand{\calO}{{\mathcal{O}}}

\newcommand{\calX}{{\mathcal{X}}}

\newcommand{\CC}{{\mathbb{C}}} 
\newcommand{\Id}{{\mathrm{Id}}} 
\newcommand{\RR}{{\mathbb{R}}} 

\DeclareMathOperator{\RE}{\mathrm{Re}}
\DeclareMathOperator{\IM}{\mathrm{Im}}
\DeclareMathOperator{\Det}{\mathrm{Det}}
\DeclareMathOperator{\Tr}{\mathrm{Tr}}
\DeclareMathOperator{\sech}{\mathrm{sech}}
\def\vartheta{\theta}

\def \sech {\mathop{\rm sech}\nolimits}

\def\Ktwo             {K}  
\def\Rtwo             {R}  
\def\Rone             {R}  
\def\Sone             {S}  
\def\THETAZ {\vartheta}
\def\THETA {\vartheta}

\def\RHO {\rho}
\def\Ss {\sigma}
\def\Ts {\mu} 
\def\Rx {s}

\def \ccomma{\raise 2pt\hbox{,}} 
\def \D {\hbox{d}}
\def\Rcubec {\mathbb{R}^3(c)}
\def\Rcube {\mathbb{R}^3}

\def\tr{\mathop{\rm tr}\nolimits}
\begin{document}

\title[Geometry of Ising persistence and the universal Bonnet-Manin
Painlev\'e VI  
]{Geometry of the Ising persistence problem and  the universal Bonnet-Manin Painlev\'e VI 
distribution}
\def\today{Accepted for publication in J. Stat. Phys. (2026)}
\date{\today}

\author*[1,2]{\fnm{Ivan} \sur{Dornic}}
\email{Ivan.Dornic@cea.fr, ORCID https://orcid.org/0000-0001-9036-8893}

\author[3,4]{\fnm{Robert} \sur{Conte}}
\email{Robert.Conte@cea.fr, ORCID https://orcid.org/0000-0002-1840-5095}

\affil[1]{
\orgdiv{Service de physique de l'\'etat condens\'e  (UMR 3680)}, 
\orgname{Universit\'e Paris-Saclay, CEA Saclay, CNRS}, 
\orgaddress{
\city{Gif-sur-Yvette}, \postcode{91191}, 
\country{France}}}

\affil[2]{
\orgdiv{Laboratoire de physique th\'eorique  de la mati\`ere condens\'ee (UMR 7600)}, 
\orgname{Sorbonne universit\'e, CNRS}, 
\orgaddress{\street{4, Place Jussieu}, \city{Paris}, 
\postcode{75252 Paris Cedex 05}, 
\country{France}}}

\affil[3]{
\orgdiv{Centre Borelli, LRC MESO}, 
\orgname{Universit\'e Paris-Saclay, ENS Paris-Saclay, CNRS}, 
\orgaddress{\street{4, avenue des sciences}, \city{Gif-sur-Yvette}, 
\postcode{91190}, 
\country{France}}}

\affil[4]{
\orgdiv{Department of Mathematics}, 
\orgname{The University of Hong Kong}, 
\orgaddress{\street{Pokfulam}, \city{Hong Kong}, 
\country{China}}}

\abstract{
We determine the full persistence probability distribution for a non-Markovian stochastic process, motivated by first-passage questions arising in interacting spin systems and allied systems.
We show that this distribution is governed by a distinguished
Painlev\'e~VI system arising from an exact Fredholm Pfaffian structure
associated with the integrable \emph{sech} kernel, $\Ksech=1/(2 \pi \cosh[(x-y)/2])$.
The universal persistence exponent originally obtained
by Derrida, Hakim and Pasquier is recovered as an asymptotic observable
and acquires a natural geometric interpretation.
In the stationary scaling regime, the persistence probability admits an exact Pfaffian decomposition into even and odd Fredholm determinants of the integrable \emph{sech} kernel.
These determinants are controlled by a unique global solution of a
second-order nonlinear ordinary differential equation, which is
identified as a particular Painlev\'e~VI equation.
The corresponding Painlev\'e~VI connection problem determines the
persistence exponent as a limiting value at infinity.

We further show that the Painlev\'e~VI system governing persistence
admits a direct geometric interpretation: the relevant solution
coincides with the mean curvature of a one-parameter family of Bonnet
surfaces immersed in $\mathbb R^3$.
A folding transformation between such surfaces singles out the
Painlev\'e~VI equation with Manin coefficients $[0,0,0,0]$, which in particular  governs
the universal persistence distribution in the symmetric Ising case.
In this framework, the persistence exponent is identified with the
asymptotic mean curvature of the associated surface.
}

\maketitle

\textit{Dedicated to Joel L. Lebowitz, the great soul of statistical physics}
\medskip

\today

\medskip


\begin{center}

\textgreek{>Agewm'etrhtos mhde`is e>is'itw }

\end{center}

(\textit{``Let no one ignorant of geometry enter here."})

[Legendary \cite{Saffrey1968} inscription written at the entrance of Aristotle's classroom in Plato's Academy.]



 

												

\noindent 
\\ 02.30.Ik -- Integrable systems
\\ 33.e17 -- Painlev\'e-type functions
\\ 02.50.Cw -- Probability theory 
\\ 02.40.Hw Classical differential geometry

\tableofcontents

\section{Introduction}

What is the  chance 
for a fluctuating  quantity  to have \textit{always} remained above its  long-term tendency  or, conversely,  the likelihood 
it \textit{first} crosses its average value
at a given time? 

The study of the first-passage properties
for a random process 
revolves around such questions,
with innumerable applications in the natural sciences \cite{RednerBook,ORReview,MOR2014}.
In the context of interacting many-body nonequilibrium systems, in particular those displaying coarsening dynamics \cite{Bray1994},
this topic has attracted a lot of interest in the 1990's 
under the name of ``{persistence}''  (see \cite{Majumdar1999,BMS2013,AS2015} for  reviews, the last two recent ones).

Probably one of the main reasons for  the upsurge of  activity for this subject, and the fascination for it,
was sparked by the discovery that 
even for the simplest possible
models  \cite{DBG1994,DHP1995,MSBC1996,DHZ1996,DHP1996}, 
the persistence probability  decays algebraically for large times,
with  
an  exponent, generally denoted $\theta$,
which does not seem to be related to  any other known static or dynamic critical exponents, although it subsumes in a simple number the  dynamics of the interwoven mosaic of growing domains.

In 1995-1996, Derrida, Hakim, and Pasquier (DHP)  obtained in \cite{DHP1995,DHP1996}  an   extraordinary exact analytical expression for a continuous family $\widehat{\theta}(q)$ of persistence exponents in a prototypical model of nonequilibrium statistical physics, the $q$-state Potts   model in one space dimension  evolving  with zero-temperature Glauber dynamics.
 If  $p_0([t_1,t_2];q)$ is the probability that the Potts spin at the origin of a semi-infinite chain has never flipped between times $t_1$ and $t_2$, they showed that when $1 \ll t_1 \ll t_2$ this probability decays as $p_0([t_1,t_2];q)\propto (t_1/t_2)^{\widehat{\theta}(q)/2}$, where 
\be
\label{DHPqformula}
\widehat{\theta}(q) =-\frac{1}{8} +\frac{2}{\pi^2}\lb \arccos{\lp \frac{2-q}{\sqrt{2} \, q}\rp}\rb^2,
\ee
 with the striking limiting value $\widehat{\theta}(2)/2=3/16$ in the  $q=2$ Ising case. 

Much  later, in 2018,  Poplavskyi and Schehr \cite{PS} obtained directly  this particular persistence exponent $3/16$,
and revealed it  was universal, since it appears in several \textit{a priori} unrelated problems, such as the determination of the   real zeros of the Kac polynomial, 
 the first-passage properties for a $D=2$ randomly diffusing field, or within the truncated orthogonal ensemble of random matrix theory.

In the present work, we  determine the full probability 
law   giving for the  $D=1$ Ising-Potts model and allied processes   their  universal persistence distribution as a distinguished  sixth Painlev\'e function,    recovering
in particular the exponent (\ref{DHPqformula})  as an asymptotic observable. 

Using instead the equivalent but more convenient viewpoint of $\pm$ Ising spins on a $D=1$ chain evolving from an arbitrary initial magnetization 
 $-1 \le m \le 1$, we show that in the appropriate stationary scaling regime the persistence probability
distribution function, $\ell \mapsto P_0(\ell;m)$,
 is  governed by a one-parameter family  of genuinely transcendental  solutions $\calH=\calH(x;\xi)$
 for a particular sixth Painlev\'e equation ($\PVI$),
\be
\label{resP0Lm}
P_0(\ell;m) = \Det{\lp \Id -\xi \Ksech^+\rp}\!\upharpoonright_{[0,\ell]}
=\exp{\lp \int_0^\ell \! \dd x 
\frac{\calH(x;\xi) - \sqrt{-\calH'(x;\xi)}}{2}\rp},
\ee
 with $\xi = \xi(m) = 1-m^2$, and $\calH'=\dd \calH/\dd x$.

The structure of this  Fredholm determinant evidences that the persistence probability is a  Pfaffian gap-spacing probability generating function, here  associated to a translation-invariant and integrable scalar kernel, the \textit{sech} kernel,
 \be
\label{Ksech}
\Ksech(x_1-x_2) \coloneqq \frac{1}{2 \pi} \, \frac{1}{\cosh{[(x_1-x_2)/2}]}, \quad x_1 -x_2 \in \RR.
\ee  
In the  formula (\ref{resP0Lm}),  $(\Id- \xi \Ksech^+)\upharpoonright_{[0,\ell]}$ is  the  integral operator
for  the  sech kernel restricted  to  its even eigenfunctions on the interval   $[0,\ell]$. For Ising spins,  the 
generating function parameter becomes a thinning parameter, 
$\xi=\xi(m) \in (0,1)$, which 
  keeps track of   the average magnetization 
  in the random initial condition.

 Our analysis further establishes 
the existence of  a unique negative and decreasing solution
$x \mapsto \mathcal H(x;\xi)$, regular at the origin, defined on the positive real axis
and globally bounded for each fixed parameter
$\xi=\xi(m)\le 1$.
Its limiting value at infinity,
\begin{equation}
\label{reskappam}
\kappa(m)\;\coloneqq\;
-\lim_{x\to+\infty}\mathcal H\!\left(x;\xi(m)\right)
= -\frac{1}{8}
+\frac{2}{\pi^2}
\left[
\arccos\!\left(\sqrt{\frac{m^2}{2}}\right)
\right]^2,
\end{equation}
\label{page2}
is entirely determined by the associated Painlev\'e~VI \emph{connection problem}
for $\mathcal H$, and therefore depends only on $m^2$.
It governs the exponential decay
$P_0(\ell;m)\propto e^{-\kappa(m)\ell/2}$ of the persistence probability on the stationary
timescale $\ell\gg1$.

The dependence on the \emph{sign} of the magnetization does not enter into the decay rate $\kappa(m)$ itself, but is entirely encoded at the probabilistic level
by a Pfaffian parity decomposition, which constitutes the first main result of this work
and mixes probabilistic and analytic ingredients.

In particular, the Painlev\'e~VI connection problem becomes singular at $\xi=1$,
 where the Fredholm determinants for the sech kernel develop  a Fisher-Hartwig singularity \cite{FH1968,DIK2013}, and this leads in the symmetric Ising case $m=0$ to the remarkable rational value $\kappa(0)/2=3/16$.
 The Derrida--Hakim--Pasquier exponent~\eqref{DHPqformula} is recovered asymptotically
from~\eqref{reskappam} by combining the analytic determination of the Painlev\'e~VI
connection constant $\kappa(\xi)$ with the probabilistic identification
$1/q=(1+m)/2$ relating Potts and Ising initial conditions on the positive branch.

\smallskip
\textit{But which Painlev\'e~VI equation?}
\smallskip

The determination of the precise Painlev\'e~VI system governing the persistence problem is a central issue of the present work, and ultimately justifies the title of the article.

Beyond the mere fact that the function $\calH$ entering the Fredholm representation~\eqref{resP0Lm} satisfies a closed second-order, second-degree nonlinear ordinary differential equation that we shall obtain by standard Tracy-Widom techniques for integrable operators, one must determine \emph{which} Painlev\'e~VI is actually involved.

This amounts to identifying the values of the four parameters
$\alpha,\beta,\gamma,\delta$, the boundary conditions selecting the relevant solution, and the nature --- classical or genuinely transcendental --- of the solution required for the persistence problem.

\medskip

A stunning outcome of the present work is that the function $\calH$ determining the persistence probability admits a direct geometric interpretation.
More precisely, $\calH$ can be identified with the mean curvature of a remarkable family of surfaces in three--dimensional Euclidean space discovered by Bonnet in 1867.

Within this framework, the associated Painlev\'e~VI equation arises as the Gauss--Codazzi compatibility condition of the moving frame, while geometry reveals nontrivial transformation properties of the persistence system.

In particular, it provides a natural explanation for the appearance of an angle $\arccos$ in the Derrida--Hakim--Pasquier exponent~\eqref{DHPqformula}, evocative of a nonlinear analogue of Buffon's needle problem.

More importantly, the geometric viewpoint leads to a folding transformation between distinct Bonnet Painlev\'e~VI surfaces, which makes it possible to express the required function $\calH=\calH(x;m)$ in terms of a function $Q=Q(t)$ solution of
\begin{equation}
\label{eqPVI0001rational}
\ddot{Q}= \frac{1}{2}\lp \frac{1}{Q} + \frac{1}{Q-1} + \frac{1}{Q-t} \rp \dot{Q}^2
- \lp \frac{1}{t} + \frac{1}{Q-1} + \frac{1}{Q-t} \rp \dot{Q},
\end{equation}
(with $\dot{Q}=\dd Q/\dd t$) corresponding to a $\PVI$ equation having  the simplest possible parameters 
\begin{equation}
\label{abcd0000}
\abcd = [0,0,0,0].
\end{equation}

A Painlev\'e~VI equation with the very same coefficients was encountered by Manin in 1995 in a completely different algebro--geometric context.
It thus appears fair to us to christen the particular solution for $m=0$ giving the full persistence distribution function
$P_0^\star(\ell)=P_0(\ell;m=0)$ with its universal decay rate
$\kappa^\star \ce \kappa(0)/2 = 3/16$ as
the \emph{Bonnet--Manin Painlev\'e~VI}, after the names of these two great geometers from the past centuries.

We now turn to a  precise mathematical statement of the  results contained in this work. 
We present them as two main theorems: the first  combines probability and analysis, the second analysis and geometry, with the overarching Painlev\'e VI structure providing the bridge between them.

\subsection{Statement of results}
\label{refsubsectionstatement}
\begin{theorem}[Pfaffian decomposition and Fredholm representation]
\label{thm:A}
In the stationary scaling regime and starting from $m$-magnetized random initial conditions, the  persistence  probability  $P_0^+(\ell;m)$ that  the Ising spin located at the origin of semi-infinite chain evolving with zero-temperature Glauber dynamics  has 
always been in the $+$ state for a length of time $\ell$
admits, 
along with its twin probability $P_0^-(\ell;m)$, an exact Pfaffian parity decomposition 
\begin{equation}
\label{eq:PfaffianCore}
P_0^+(\ell;m)= P_0^-(\ell;-m) 
=\frac{1+m}{2}\,\frac{\mathcal D^+ + \mathcal D^-}{2}
+ \frac{1-m}{2}\,\frac{\mathcal D^+ - \mathcal D^-}{2}.
\end{equation}
Here, $\calD^\pm= \mathcal D^\pm(\ell;\xi)$, with $\xi=\xi(m)=1-m^2$, are the Fredholm determinants
for the even and odd parts of the thinned sech kernel restricted to $[0,\ell]$:
\begin{equation}
\mathcal D^\pm(\ell;\xi)
\coloneqq
\Det\!\bigl(\Id-\xi\,K_{\mathrm{sech}}^\pm\bigr)\!\upharpoonright_{[0,\ell]} \, =\exp{\lp \int_0^\ell \! \dd x 
\frac{\calH(x;\xi) \mp \sqrt{-\calH'(x;\xi)}}{2}\rp}.
\label{def-calD}
\end{equation}
Both are determined by the unique solution $\calH(x)=\calH(x;\xi)$ regular at the origin,   negative, decreasing,
 and globally bounded on $\mathbb R_+$ for the Cauchy problem
 \begin{eqnarray}
& &
\left\lbrace
\begin{array}{ll}
\displaystyle{\lp \frac{\calH''}{2 \calH'} +
\coth{x} \rp^2 + \frac{1}{(\sinh{x})^2} \frac{\calH^2}{\calH'} +2 (\coth{x}) \calH+ \calH' = \fs{4}, \quad x >0,}\\ 
{}\\ \label{ODEHThmA}
 \displaystyle{\calH(0^+)=  -\xi  \quad \text{and} \quad \calH'(0^+)= -
 \calH^2(0^+), \quad  x\searrow 0, \quad 0 < \xi \le 1.}
\end{array}
\right.
\end{eqnarray}
In  the  scaling regime, the persistence distribution $P_0^+(\ell;m)$ for this Ising spin
is related to $p_0([t_1,t_2];q)$, the one  for a Potts 
spin,  through the identification 
\begin{equation}
P_0^+(\ell;m)=\frac{1}{q}\,p_0([t_1,t_2];q), \quad \ell = \log(t_2/t_1) >0, \quad \frac{1}{q}= \frac{1+m}{2}, 
\end{equation}
 valid for arbitrary $\ell>0$ and $m\ge 0$, which implies that  
 $\kappa(m) = -\lim_{x \to \infty}\calH(x;\xi(m))$,
 \label{pageThmA}
 the negative limiting value at infinity attained by the solution to the ODE (\ref{ODEHThmA}), 
 gives back the half-space DHP persistence exponent:
\be
P_0^+(\ell;m)  \propto e^{-\kappa(m)\ell/2}, \quad \ell \gg 1 \quad \Longrightarrow \quad  p_0([t_1,t_2];q) \propto (t_1/t_2)^{\widehat{\theta}(q)/2}, \quad t_1/t_2 \ll 1.
\ee
\label{persistence-scaling-limit}
\end{theorem}

We recall that, as first emphasized by Derrida, Hakim and Pasquier in \cite{DHP1996}, the persistence problem must be formulated on a semi-infinite chain in order for the underlying Pfaffian structure to emerge: only in this geometry do the two half-lines on each side of a spin which never flips  evolve independently, leading to a factorization of the persistence probability on the infinite chain and to a persistence exponent which is twice that obtained in the semi-infinite setting.

As for the equivalence between $q$-state Potts spins and $\pm$ Ising ones on a chain evolving with zero-temperature Glauber dynamics from random initial conditions, it is a well-known probabilistic fact, see, e.g. \cite{LiggettBook}.
Indeed, by declaring that one Potts color, occurring with probability $1/q$, corresponds to an Ising spin $+$, while the remaining $q-1$ colors grouped together correspond to an Ising spin $-$, one obtains an Ising chain with magnetization $m$ such that $1/q=(1+m)/2$.
In the following, we shall exclusively use the Ising language for the persistence problem, notably because this correspondence provides a natural extension to non-integer values of $q=2/(1+m)$.

Equation~(\ref{eq:PfaffianCore}) makes transparent how the Pfaffian structure propagates the Bernoulli weights $(1\pm m)/2$ of the initial condition through the stationary scaling regime, by deforming them into a diagonal superposition of the even and odd Fredholm determinants. Of course, this formula is equivalent  to $P^\pm_0 = \calD^+ \pm m \calD^-$, showing that the dynamics effectively separates the two parity sectors.
Summing over the two  possible  states a  spin can stay without flipping, one recovers
for the   persistence probability 
distribution  the expression
(\ref{resP0Lm}) 
as a  single Fredholm Pfaffian determinant,
\begin{equation}
P_0(\ell;m)=P_0^+(\ell;m)+P_0^-(\ell;m)
=\mathcal D^+(\ell;\xi), \qquad \xi=1-m^2,
\end{equation} namely the one for the even part of the sech kernel. 

\begin{theorem}[Painlev\'e~VI structure of persistence and global solution of Bonnet surfaces]
\label{thm:B}

\smallskip

\noindent
\textbf{(i) Geometric interpretation.}
There exists a one-parameter family of immersed surfaces
$\Sigma_\xi \subset \mathbb R^3$ admitting conformal coordinates
$(z,\bar z)$ such that the function $\calH=\mathcal H(x;\xi)$ appearing in
Theorem~\ref{thm:A} coincides with the mean curvature of $\Sigma_\xi$.
The surfaces $\Sigma_\xi$  belong to the class of Bonnet surfaces, for which both the mean curvature function and the metric depend on the sole
real variable $x=\Re z$.
This yields therefore a global solution of the Gauss--Codazzi equations for $\Sigma_\xi$.

\smallskip

\noindent
\textbf{(ii) Painlev\'e~VI avatars and transformations.}
The quantities associated with the persistence probability
($\mathcal H$, $\tfrac12(\mathcal H\mp\sqrt{-\mathcal H'})$,
and appropriate combinations of the logarithmic derivatives of $\mathcal D^\pm$)
satisfy intertwined Painlev\'e~VI equations.
These equations are birationally or algebraically related, possess contiguous
monodromy exponents, and admit a natural geometric interpretation on $\Sigma_\xi$,
as summarized in Table~\ref{table:identity}.

\smallskip

\noindent
\textbf{(iii) Folding transformation and Manin parameters.}
A folding transformation relating distinct Bonnet surfaces
singles out a distinguished Painlev\'e~VI equation with parameters
$$
\abcd=[0,0,0,0],
$$
previously identified by Manin.
This particular Painlev\'e~VI equation governs the universal persistence
distribution $P_0^\star(\ell)=P_0(\ell;0)$.

\smallskip

\noindent
\textbf{(iv) Asymptotic curvature and persistence exponent.}
For all $\xi=\xi(m)$, the persistence exponent $\kappa(m)$ is equal to
minus the asymptotic mean curvature of $\Sigma_\xi$ at infinity.
This limit corresponds to the unique umbilic point of $\Sigma_\xi$,
where the two principal curvatures become equal.
\end{theorem}

This last sentence is the reason why the persistence exponent is here denoted 
$\kappa$ (like a geometric curvature) and not $\theta$ as usual.
Another reason is the standard notation $\theta$ for the monodromy exponents
of $\PVI$ which determine the four parameters of the generic $\PVI$ equation,
\be
\label{tjsq}
(\va^2, \vb^2, \vc^2, \vd^2) \coloneqq (2 \alpha, -2 \beta, 2 \gamma,1 -2 \delta), 
\ee
a convention which respects the parity under any permutation of the locations 
$j=\infty,0,1,t$ of the four defining singularities of $\PVI$.

This implies three important features for the present universal probability law.
\label{Three-features}
\begin{enumerate}
    \item It is \textit{not} related by any birational or algebraic transformations to the ``usual'' Picard case \cite{Picard}
with
 all its monodromy exponents $\la \vta_j \ra$ equal to zero
 (hence with coefficients $\abcd=[0,0,0,1/2]$),
and whose general  (two-parameter initial conditions dependent)  solution  is   known explicitly;
    \item It is transcendental in the 19-th century sense (nonalgebraic dependence in the two constants of integration);
    \item It is not reducible to classical hypergeometric
    solutions of $\PVI$ already encountered in a random matrix 
    context \cite{FWNagoya2004}.
\end{enumerate}

\begin{table}[ht]
\caption{Correspondence between 
persistence,
sixth Painlev\'e function,
and the two principal curvatures $\kappa_1, \kappa_2$ of the Bonnet surface.
The successive columns display:
Fredholm determinants for the sech kernel
as defined in (\ref{def-calD}), signed
monodromy exponents $\vabcd$ for each associated $\PVI$,
 geometric quantity.
} 
\centering 
\begin{tabular}{c c c c} 
\hline\hline 
Fredholm determinant &
Monodromy exponents &
Geometric quantity &
 \\  [0.5ex] 
   \hline $\calD^\pm/\calD^\mp$    & $\la 0,0,0,1\ra $                                         & $\calK=\kappa_1 \kappa_2$ \, (total curvature) & 
\\ \hline $\calD^+ \times \calD^-$ & $\la \frac{1}{2},0,0,\frac{1}{2}\ra $                     & $\calH=\frac{\kappa_1+\kappa_2}{2}$ 
\, (mean curvature)& 
\\ \hline $\calD^\pm$              & $\la \frac{1}{4},\frac{1}{4},\frac{1}{4},\frac{1}{4}\ra $ & $\calL=\frac{\kappa_1-\kappa_2}{2}$ \, (skew curvature) & 
\\ \hline 
\end{tabular}
\label{table:identity} 
\end{table}


\subsection{Relation to existing literature}
\label{refsubsectionexistinglit}

The persistence problem can be viewed as a first-passage question for a
non-Markovian stochastic process, a perspective going back to early studies
of stationary Gaussian processes.
In the context of interacting many-body systems, this viewpoint gained
prominence in the 1990s through coarsening dynamics, where persistence
quantifies the probability that a single degree of freedom never changes
its state as macroscopic domains form and evolve.

The intrinsically non-Markovian nature of this problem was quickly
recognized as a major obstacle to exact calculations, making the result
obtained by Derrida, Hakim and Pasquier a genuine tour de force:
their exact expression for the persistence exponent in the one-dimensional
Potts model remained for a long time an isolated achievement.

More than twenty years later, a decisive advance was made by Poplavskyi and Schehr~\cite{PS}, who computed
directly the Ising persistence exponent $3/16$ and demonstrated its
universality by showing that it appears in several \emph{a priori} unrelated
problems. These systems are Gaussian (at least asymptotically) and share
the same non-Markovian stationary correlator
\begin{equation}
\label{Csech}
C(T_2-T_1) \ce \mathbb{E}\!\left[\mathcal X(T_1)\mathcal X(T_2)\right]
= \sech(T_2-T_1).
\end{equation}
Their work firmly established the universal nature of the exponent, but did
not address the full persistence distribution nor its analytic structure.

Rigorous asymptotic results for Fredholm Pfaffians associated with
integrable kernels were later obtained by FitzGerald, Tribe and
Zaboronski~\cite{FTZ2022}. Their analysis provides a precise mathematical
framework for large-gap asymptotics and clarifies the role of
Fisher--Hartwig singularities in Pfaffian point processes, but does not
identify the nonlinear differential equations governing the corresponding
probability laws.

Finally, in an unpublished manuscript~\cite{Dornic2018} motivated by the results of \cite{PS}, the first author
observed that the Ising persistence probability admits a Fredholm determinant
representation involving the integrable sech kernel. However, neither the
solution of the associated Painlev\'e~VI connection problem nor the full
underlying geometric structure  in terms of intertwined $\PVI$ were identified at that stage.

The key advance of the present work is to recognize that the persistence
probability can be identified with a standard object of random matrix
theory: a conditional gap-spacing probability endowed with a Pfaffian structure for the  integrable kernel
$\Ksech$, thus essentially the sech correlator
(\ref{Csech}).
 This identification provides direct
access to the full and universal persistence distribution and to its governing
Painlev\'e~VI system, since for the particular non-Markovian Gaussian stationary process with the correlator (\ref{Csech}), we prove that its conditional first-passage probability distribution function is 
\begin{equation}
\mathbb{P}\!\left[
\inf_{0<T<\ell} \mathcal X(T) > 0 \,\middle|\, \mathcal X(0)=0
\right] = \Det\!\bigl(\Id-K_{\mathrm{sech}}^+\bigr)\!\upharpoonright_{[0,\ell]} \equiv P^\star_0(\ell).
\end{equation}

\subsection{Outline of the paper}

The paper is organized so as to progressively unfold the integrable structure underlying persistence. 

We first establish a probabilistic and analytic formulation of the persistence problem in terms of Fredholm Pfaffians and a closed nonlinear differential system. This leads to the proof of Theorem A, which isolates the Pfaffian parity decomposition and the associated second-order second-degree nonlinear ordinarry differential equation (ODE) satisfied by the logarithmic derivatives of the Fredholm Pfaffians determinants. 

We then reinterpret this system geometrically, identifying the relevant Painlev\'e VI functions with the mean curvature of Bonnet surfaces in 
$\RR^3$.  This viewpoint culminates in Theorem B, which clarifies the rôle of folding transformations, monodromy data, and the distinguished Manin 
 parameters $\abcd=[0,0,0,0]$.

Several technical complements, asymptotic analyses, and comparisons with existing results are gathered in the remaining sections and appendices.

\section{Integrable structure of  persistence probability: proof strategy}
\label{section-Proofs}

The proofs of Theorems~\ref{thm:A} and~\ref{thm:B} rely on a combination of
probabilistic, analytic, and geometric arguments.
Rather than presenting them in a purely axiomatic manner, we have organized the
presentation according to the underlying integrable structure, progressing from
general properties of Pfaffian point processes to the specific features of
the persistence kernel.

The argument proceeds through the following steps:
\begin{itemize}
    \item 
    In Section~\ref{sectionDPP}, we recall the basic properties of
    determinantal and Pfaffian point processes, coming from from a family of probability convolution kernels defined
    on the real line or the unit circle, and which naturally arise in the
    persistence problem.

    \item 
    In Section~\ref{section-closed-differential-system}, we consider the
    integrable kernel~\eqref{defKtanu} with a general parameter $\vartheta$
    and derive a closed differential system governing the associated
    Fredholm determinants.

    \item 
    Section~\ref{section-integration} is devoted to the explicit integration
    of this system and to its formulation in terms of Painlev\'e functions.

    \item 
    In Section~\ref{section-Folding-symmetries}, we specialize to the value
    $\vartheta=1/2$ relevant for persistence and show how the resulting
    system can be reduced, via a folding symmetry, to the particular
    Painlev\'e~VI equation identified by Manin.

    \item 
    Finally, Section~\ref{section-transcendence} establishes the genuinely
    transcendental nature of the resulting universal distribution.
\end{itemize}

\subsection{Fredholm or Pfaffians determinants and  gap probabilities}
\label{sectionDPP}

Every member of the family of translation-invariant  even kernel functions  
of the generic form (\ref{defKtanu})  defines a point process either  on the unit circle (when $\nu^2<0$) or on  the real line (when $\nu^2>0$). In the latter case, note that $\nu$ is thus an arbitrary real scaling factor, while in the former one it is of course constrained by $2 \pi$. 
For concrete applications, one has typically $\nu=\ii$ or $\nu=1$ or $\nu=2$ with the respective kernels of reference $K_N$ or $K_\ta$: both choices have pros and cons \cite[p 66]{TWFred1994}, and we shall use both, depending whether one takes the rightmost endpoint of the symmetric interval $[-T,T]$ as the independent variable, or its length $\ell=2 T$. 
Since, the ODEs obeyed by the associated resolvents  contain 
terms $\coth(2 \nu T) \equiv \coth(\nu \ell)$,  it should be clear which choice is in use from the corresponding equations. 

We shall everywhere follow the terminology of Tracy-Widom \cite{TWGOSE1996} (section II):
``If $K(x, y)$ is the kernel of an integral
operator $K$ then we shall speak interchangeably of the determinant \textit{for} $K(x,y)$ or $K$, and the determinant \textit{of} the operator $\Id - K$.''

In order to simplify the notation in the present section,
we simply denote  $K(x)$ any kernel $ K_\vartheta(\nu x)$. Each of   this continuous, bounded, and  symmetric kernel when it  acts on some compact interval $I=[T_1,T_2]$ defines
a linear operator $K\!\upharpoonright_{[T_1,T_2]}$ on $L^2(I,\dd x)$, which we also write as
$K\!\upharpoonright_I$,   or even 
 just as  $K_\upharpoonright$. This operator is  always locally trace-class for any finite interval length, and its eigenvalues  are simple, discrete, non-negative, bounded, 
and ordered as
\be
1 > \lambda_0(I) > \lambda_1(I) > \dots >0.
\ee 

The Fredholm determinant generating  function is obtained by
multiplying the kernel by $\xi$,
\be
\label{defdetKgf}
D([T_1,T_2];\xi)= \sum_{n=0}^{\infty}\frac{(-\xi)^n}{n!}\int_{T_1}^{T_2}\! \dd x_1 \dots \int_{T_1}^{T_2}\! \dd x_n
\underset{1 \le j,k \le n}{\mathrm{det}}{\lb K(x_j-x_k)\rb}
\ee
and it is equal to the spectral determinant
\be
\label{defdetKsp}
D([T_1,T_2];\xi)=\Det{\lp \Id-\xi K\rp}\!\upharpoonright_{[T_1,T_2]}\, =\prod_{n }\lp 1-\xi \lambda_n(I) \rp.
\ee

By translation invariance of the kernel, it suffices to consider this determinant on the symmetric interval $[-T,T]$ of length $\ell=2 T$ around the origin, in which case 
the even and odd parts of the kernel 
\be
K^\pm(x,y) \coloneqq \frac{1}{2}\lp K(x,y) \pm K(-x,y) \rp \equiv  \frac{1}{2}\lp K(x-y) \pm K(x+y) \rp,
\ee
select the even and odd 
square-integrable eigenfunctions:
\be
\int_{-T}^T\! \dd y K^{\pm}(x,y)f_n(y)=
\int_0^T \! \dd y (K(x-y) \pm K(x+y))f_n(y)=\lambda_{n}(T) f_n(x), \quad 
n \, \mathrm{even}/\mathrm{odd}.
\ee
In terms of the resolvent operator $R$ for  $\xi K_\upharpoonright$, defined as usual through 
\be
\label{defR}
\Id + R = (\Id-\xi K_\upharpoonright)^{-1},
\ee
the corresponding Fredholm determinants are \cite[Eq.~(30)]{TWGOSE1996}
\be
\label{def-D-calD}
D^\pm([-T,T];\xi) \coloneqq \Det{\lp \Id-\xi K^\pm\rp}\!\upharpoonright_{[-T,T]}\,
=\exp{\lp -\int_0^T \! \dd x \lb R(x,x;\xi) \pm R(-x,x;\xi) \rb \rp}.
\ee
Of course, the product of these two,  for short $D=D^+.D^-$, gives for a symmetric interval the Fredholm determinant (\ref{defdetKsp})  for $\xi K_\upharpoonright$, which is thus such that
\be
\label{defD2T}
D([-T,T];\xi) = \prod_{n \textrm{ even and odd}}(1 -\xi \lambda_n(T))=
\exp{\lp -2 \int_0^T \! \dd x \,  R(x,x;\xi)\rp}.
\ee

 For any   $\xi K\upharpoonright_{[-T,T]}$ operator on a symmetric interval with an even kernel function $K(x-y)=K(y-x)$, the resolvent kernel functions at coincident and opposite endpoints
obey the symmetry relations
\be
\label{RSsym}
R(x,x;\xi) = R(-x,-x;\xi), \quad R(x,-x;\xi)=R(-x,x;\xi), \quad -T < x < T,
\ee
along with a crucial differential constraint
which was first found by Gaudin  (but never published!) for the sine kernel, but which remains valid \cite[Appendix A16]{MehtaRMTBook} for any even-difference kernel function
\be
\frac{\dd}{\dd x}R(x,x;\xi)= 2 \lb R(-x,x;\xi)\rb^2, \quad -T < x < T.
\ee

It is therefore  sufficient to consider these two resolvent functions for $0<x<T$, and even more conveniently \cite[p.65]{TWFred1994} to view them as functions of  the (variable) rightmost endpoint $T$ of the interval. Therefore we  abbreviate and define
for all $T>0$   
\be
\label{defRTR0}
R(T) \coloneqq \lim_{x \to T^-}R(x,x;\xi), \quad R(0^+) = \xi K(0),
\ee
the second equality coming from the Neumann expansion of the resolvent (\ref{defR}),
while from Gaudin's relation (\ref{RSsym}) evaluated for $x \to T^-$, we  set  
\be
\label{defST}
S(T) \coloneqq \lim_{x \to T^-}R(-x,x;\xi) \quad \Longrightarrow \quad S(T)= \sqrt{R'(T)/2}, \quad S(0^+) = R(0^+)
\ee
after choosing the $+$ sign when solving backwards for $S$ in (\ref{RSsym}).
For both these functions, the thinning parameter is now hidden in their value   at the origin $T=0$, along with the  normalization chosen for the \textit{density} $\rho_1= K(0)$ (see (\ref{defrhon}) below) of the stationary point process generated by the translation-invariant kernel. 

All the determinants $D,D^+$, and $D^-$ can thus be reconstructed 
after quadratures from the sole knowledge of the  function $S(T)$, 
\be
\label{resStoD2det}
S(T) = \frac{1}{2}\sqrt{-\frac{\dd^2}{\dd T^2}\log{D([-T,T];\xi)}}.
\ee

Finally, we recall the definitions of the 
``inclusive" and ``exclusive''  correlation functions. The first are given by  
\be
\label{defrhon}
\rho_n(x_1,x_2,\dots,x_n) =  \underset{1 \le j,k \le n}{\mathrm{det}}{\lb K(x_j-x_k)\rb},
\ee
so that the  value for $\xi=1$ of the Fredholm determinant (\ref{defdetKgf}) when expanded, namely
\be
D([T_1,T_2];1)=1 -\int_{T_1}^{T_2}\! dx_1 \, \rho_1 + \frac{1}{2!} \int_{T_1}^{T_2}
\int_{T_1}^{T_2} \! dx_1 dx_2 \, \rho_2(x_1,x_2) - \cdots,
\ee
is the \textit{gap-spacing probability}, i.e. the probability that the interval $[T_1,T_2]$ is void of any points for this stationary point process with (constant) average density $\rho_1=K(0)$.

As for the second  (or J\'anossy) correlation functions, they are defined by
\be
\omega_n(x_1,x_2,\dots,x_n;\xi) =  D([T_1,T_2];\xi). \underset{1 \le j,k \le n}{\mathrm{det}}{\lb R(x_j,x_k;\xi) \rb}.
\ee
In particular, these exclusive correlation functions   yields  for the (unthinned) point process generated by the kernel on $I=[T_1,T_2]$  
the probability
\be
\label{defPJm}
\mathbb{P}\lb \mathrm{exactly}\, m \, \mathrm{points \, in} \, I\rb = 
\Det{\lp\Id -K\!\upharpoonright_I\rp} \, \frac{1}{m!}\!\int_{T_1}^{T_2}\!\dd x_1 \dots \int_{T_1}^{T_2}\!\dd x_m 
\underset{1 \le j,k \le n}{\mathrm{det}}{\lb R(x_j,x_k)\rb}.
\ee
After division by the Fredholm determinant for $K$ on
$[0,T]$, which from (\ref{defD2T})  is   $e^{-\int_0^T R}$  (note the absence of factor $2$ this time), the case $m=1$ in  (\ref{defPJm}) shows that $\overline{\omega}_1(T) \dd T \ce R(T)\dd T$ is the \textit{conditional} probability to find the first point in $(T,T+\dd T)$ (and none on $[0,T]$), while 
the case $m=2$ is the  conditional probability density  to find a point at $0$, another one at $T$, and none in between. Using the symmetries (\ref{RSsym}), this joint conditional density function is
therefore
\be
\overline{\omega}_2(T)= R^2(T) -S^2(T),
\ee
a relation which turns out to give the square of the local skew curvature $\calL$ for
our Bonnet-$\PVI$ surfaces, cf. Table \ref{table:identity}, hence their global Willmore energy upon integration.


\subsection{Integrable structure and closed ODE system obeyed by the resolvent kernels}
\label{section-closed-differential-system}

Two known kernels in random matrix theory
have given rise to 
Fredholm determinants and gap-spacing probability expressed by
some $\PVI$ function.
The first one is the famous circular unitary ensemble $\text{CUE}_N$ introduced by Dyson
\cite{Dyson1962},
\be
 \label{defKN}
K_N(x_1-x_2) \coloneqq \frac{1}{2 \pi} \, \frac{\sin{\lb N(x_1-x_2)/2 \rb}}{\sin{(x_1-x_2)/2}},  \quad N \, \mathrm{integer}=1,2,\dots, \quad -\pi <x_1-x_2 \le \pi,
\ee
which gives the eigenvalues of a random unitary matrix $U(N)$ distributed according to the Haar measure.

The second one has been considered by Nishigaki  \cite{Nishigaki1998} to study
some ``critical" random matrix ensembles:
\be
  K(x_1-x_2) = \frac{a \sin \lb \pi(x_1-x_2)\rb}{\pi \sinh\lb a(x_1-x_2)\rb}, \quad x_1-x_2 \in \RR.
\ee

For the sech kernel of central interest for persistence, since 
 \be
 \frac{1}{\cosh{(x/2)}}= \frac{2 \sinh{(x/2)}}{2 \sinh{(x/2)} \cosh{(x/2)}}= \frac{2 \sinh{(x/2)}}{\sinh{(x)}},
\ee
a natural observation
  is the identity in the complex plane 
provided one waives the constraint $N$ integer,
\be
\label{NinC}
\Ksech(x_1-x_2) = K_N (2 \ii(x_1-x_2)) \Big{|}_{N=1/2}.
\ee

We therefore consider a common extrapolation to these two
kernels acting either on the real line or on the unit circle,
also containing the sech kernel.
Such an extrapolation is the following  
 family of translation-invariant even kernels,
\be
\label{defKtanu}
  K_{\vartheta}(\nu(x_1-x_2)) \coloneqq  \rho_\vta \, \frac{\sinh{\lb \vartheta \, \nu(x_1-x_2)\rb}}{\vartheta \sinh{\lb \nu(x_1-x_2)\rb}} , 
   \quad \rho_\vartheta >0, \quad  \vartheta^2, \nu^2 \in \RR.
\ee

In this section, 
we establish the system of ODEs associated to this kernel.

Let us first explain the roles of the two parameters $\nu^2$ and $\theta$.

According to the sign of $\nu^2$,
the determinantal point process defined by this kernel 
is either on the circle ($\nu^2<0$) or on the real line ($\nu^2\ge 0$).
In both cases, the normalization factor $\rho_\vta$ is determined 
by the condition that $\Kta$ always be an even probability distribution function: $1=\int \Kta$.

As to the real parameter $\theta$,
it will appear to be the unique monodromy parameter characterizing the Painlev\'e $\PVI$ occurring for Bonnet surfaces.

Note that, when the kernel acts on the real line ($\nu^2>0$),
the parameter $\nu$ can be scaled out. Two choices we shall employ are $\nu=1$ or $\nu=2$ for reasons which will be made clear below. In all cases,
one must restrict $-1<\theta<1$ to ensure the probability normalization $\int_{\RR} \Kta=1$:
\be
\label{defKta}
\Kta(x_1-x_2) \coloneqq \frac{\cot{( \pi \ta/2)}}{\pi} \, \frac{\sinh{[\ta(x_1-x_2)]}}{\sinh{(x_1-x_2)}}, \quad 
\,-1 < \ta < 1, \quad x_1-x_2 \in \RR,
\ee
It is also worth noticing the well-defined  pointwise limit  $\ta \to 0$ for  the latter kernel,
\ba
\label{defK0}
 K_0(x_1-x_2) &\coloneqq& \frac{2}{\pi^2} \, 
 \frac{x_1-x_2}{\sinh{[(x_1-x_2)]}}  = \lim_{\ta \to 0} \Kta(x_1-x_2),
\ea
which exists thanks to the probability normalization, and which will be shown to yield the Bonnet-Manin $\PVI$.

All these kernels are  the simplest  members of a class of  \textit{exponential variants} of  integrable operators $K\!\upharpoonright_I$  introduced and 
studied
by Tracy and Widom \cite{TWFred1994}. These operators  have kernels of the form
\begin{eqnarray}
& & {\hskip -10.0 truemm}
\label{defKchi}
\Ktwo(x,y) \chi_I(y),\ 
\Ktwo(x,y)
=\frac{\phi(x) \psi(y) - \phi(y) \psi(x)}{e^{2\nu  x}-e^{2\nu  y}}\ccomma
\label{eqKernelWronskian}
\end{eqnarray}
in which 
$\chi$ is the characteristic function of some interval $I=[T_1,T_2]$
(or of a union of disjoint intervals),
and the vector $(\phi,\psi)^{\rm t}$ is the general solution for $z \in \CC$
of the $2 \times 2$ matrix first-order ODE
\begin{eqnarray}
& & {\hskip -10.0 truemm}
\frac{\D}{\D z} 
\begin{pmatrix} \phi \cr \psi \cr \end{pmatrix}
			=M
\begin{pmatrix} \phi \cr \psi \cr \end{pmatrix},
M=\begin{pmatrix} a & b \cr c & d \cr \end{pmatrix}.
\label{eqMatrix-for-kernel}
\end{eqnarray}


If one defines for $x \in I$ the two auxiliary functions of one variable \cite[Eq.~(5.5)]{IIKS1990},
\begin{eqnarray}
\label{defPhiPsi}
& &
\left\lbrace
\begin{array}{ll}
\displaystyle{
\Phi(x;[T_1,T_2])=\bra{x} (\Id-K)^{-1}\ket{\phi},\
}\\ \displaystyle{
\Psi(x;[T_1,T_2])=\bra{x} (\Id-K)^{-1} \ket{\psi},
}
\end{array}
\right.
\end{eqnarray}
using the convenient bra-ket notation,
then the scalar resolvent $\Rtwo(x,y)$, i.e. the kernel of $(\Id-K)^{-1}K\! \upharpoonright_I$,
can be written 
like the scalar kernel $\Ktwo(x,y)$,
\begin{eqnarray}
& &
\forall x \in [-T,T],\
\forall y \in [-T,T],\
\Rtwo(x,y)
=\frac{\Phi(x) \Psi(y) - \Phi(y) \Psi(x)}{e^{2\nu  x}-e^{2 \nu  y}}\cdot
\label{eqdefPhiPsi}
\end{eqnarray}

Moreover, if the elements $(a,b,c,d)(z)$ of $M$ obey a linearizable ODE system,
then, as proven by 
Its, Izergin, Korepin, Slavnov \cite{IIKS1990} and
Tracy and Widom \cite{TWFred1994},
there exists a closed system of nonlinear integrable partial differential equations (PDE),
whose dependent variables are scalar products built from the resolvent,
and whose independent variables are the end points of the interval(s).
This PDE system is the generalization for this class of integrable kernels of the one found for the sine kernel by JMMS \cite{JMMS1980}.

In our case, the family of probability-convolution kernels (\ref{defKtanu}) is generated simply by taking a constant matrix $M$ conjugated to one with eigenvalues $\nu(1 \pm \theta)$
\begin{eqnarray}
& & 
P^{-1} M P=\nu \begin{pmatrix} 1 + \theta & * \cr 0 &  1 - \theta \cr \end{pmatrix},\
\label{eqMatrix-vartheta}
\end{eqnarray}
where we keep a   Jordan form with a non-zero upper-right element  so as to have always two linearly independent solutions $\phi, \psi$ even in the degenerate case $\vta \to 0$.

 Compared to  the previous section (\ref{sectionDPP}),  we simplify  the notation for the resolvent kernel $R(x,y;\xi)$ to $R(x,y)$, and we continue to denote
 by  $R(T)$ and $S(T)$
the respective limits as $x\to T^-$ of the resolvent kernel at coincident and opposite endpoints: the thinning parameter $\xi$ is taken into account by  the value at the origin $T=0$ of these two functions $R$ and $S$, cf. (\ref{defRTR0}), (\ref{defST}),  and (of course !) it  does not appear explicitly in the ODEs they obey we are going to establish: since the interval $I=[T_1,T_2]=[-T,T]$ depends on only one variable,
the PDE system degenerates to an ODE system, integrable by construction.

This ODE system is six-dimensional,
and its dependent variables are
the six fields $\Rone(T)$, $\Sone(T)$  defined through (\ref{defST}), 
and the four scalar products 
$p_j(T),q_j(T)$, $j=1,2$ which are obtained by taking in (\ref{defPhiPsi}) the limits $x \to T_{j}$, with $T_1=-T$, $T_2=T$,
\begin{eqnarray}
& &
\left\lbrace
\begin{array}{ll}
\displaystyle{
p_1(T)=\bra{-T} (\Id-K)^{-1} \ket{\phi},\ p_2(T)=\bra{T} (\Id-K)^{-1} \ket{\phi},\
}\\ \displaystyle{
q_1(T)=\bra{-T}  (\Id-K)^{-1} \ket{\psi},\ q_2(T)=\bra{T} (\Id-K)^{-1} \ket{\psi},
}
\end{array}
\right.
\label{eqdefpjqj}
\end{eqnarray}
which obey a closed differential system
\cite[Eqs.~(8.16)--(8.17)]{IIKS1990}.

First, the limits $(x,y) \to (T,-T),(T,T),(-T,-T)$
in (\ref{eqdefPhiPsi}) generate three algebraic (nondifferential) equations
\begin{eqnarray}
& &
\left\lbrace
\begin{array}{ll}
\displaystyle{
(x\to +T,y \to -T):\ S(T)=\frac{p_1 q_2 - p_2 q_1}{e^{2\nu  T}-e^{-2 \nu T}},
}\\ \displaystyle{
(x\to +T,y \to +T):\ 2 \nu e^{+2 \nu T} \Rone(T)=\hbox{polynomial}(p_j,q_j,p_j',q_j'),\
}\\ \displaystyle{
(x\to -T,y \to -T):\ 2 \nu e^{-2 \nu T} \Rone(T)=\hbox{polynomial}(p_j,q_j,p_j',q_j'),\
}
\end{array}
\right.
\end{eqnarray}
with $p_j(T),q_j(T)$, $j=1,2$ defined in (\ref{eqdefpjqj}).

Secondly, the derivatives of $p_j$ and $q_j$ 
\begin{eqnarray}
& &
\left\lbrace
\begin{array}{ll}
\displaystyle{
\frac{\D q_1}{\D T}=   -a q_1 - b  p_1 +2 S q_2,\
\frac{\D p_1}{\D T}=   -c q_1 - d p_1 +2 S p_2,} \\ \displaystyle{
\frac{\D q_2}{\D T}=\ \ a q_2 + b  p_2 +2 S q_1,\
\frac{\D p_2}{\D T}=\ \ c q_2 + d p_2 +2 S p_1,
}
\end{array}
\right.
\label{eq-diff-pjqj} 
\end{eqnarray}
do not introduce any new function in the system.

To summarize,
the closed differential system is made of
the four first order ODEs (\ref{eq-diff-pjqj}) for $p_j(T),q_j(T)$
and
three algebraic equations between $(p_j,q_j,\Rone,\Sone)$,
\begin{eqnarray}
& & {\hskip -13.0 truemm}
\left\lbrace
\begin{array}{ll}
\displaystyle{
p_1 q_2 - p_2 q_1 -2 \sinh(2 \nu T) \Sone=0,
}\\ \displaystyle{
b(p_1^2+p_2^2)-c(q_1^2+q_2^2)+(a-d) (p_1 q_1 + p_2 q_2)
- 4 \nu \cosh(2 \nu T) \Rone -4 \sinh(2 \nu T) \Sone^2 =0,
}\\ \displaystyle{
b(p_2^2-p_1^2)-c(q_2^2-q_1^2)+(a-d) (p_2 q_2 - p_1 q_1)
- 4 \nu \sinh(2 \nu T) \Rone=0.
}
\end{array}
\right.
\label{eq-alg-pjqjRS} 
\end{eqnarray}

By derivation \textit{modulo} (\ref{eq-diff-pjqj})
of the three equations (\ref{eq-alg-pjqjRS})
algebraic in $p_j,q_j$, $\Rone,\Sone$,
one generates three more equations
algebraic in $p_j,q_j$, $\Rone,\Sone,\Rone',\Sone'$ (with ${}'=\dd/\dd T$).

The further algebraic elimination of the four derivatives $p_1', q_1', p_2', q_2'$
and of the four fields $p_1, q_1, p_2, q_2$ generates two ODEs only involving $\Rone$ and $\Sone$, whose coefficients, as expected, just involve the coefficients of the matrix $M$ through its constant spectral invariants $\tr{M}$ and $\det{M}$.

The first of these ODEs is the one of Gaudin (\ref{eqodeGaudin}) 
\begin{eqnarray}
& &
\Rone'-2 \Sone^2=0,
\label{eqodeGaudin} 
\end{eqnarray}
which in this framework is recovered as a \textit{consequence} of the above differential system. Indeed, the sole condition $\tr{M}=a+d=\nu$ coming from the specific form (\ref{eqMatrix-vartheta})  enforces that the kernel is an even-difference one $K(x,y)=K(x-y)=K(y-x)$. 

The second one can be nicely written as the sum of four squares,
\begin{eqnarray}
\lb\sinh(2 \nu T) \Rone\rb'^2 - {\lb\sinh(2 \nu T) \Sone\rb'}^2 -\lb 2 \nu \sinh(2 \nu T) \Rone\rb^2 + \lb 2 \nu \vartheta \sinh(2 \nu T) \Sone\rb^2=0.
\label{eqode1RS}
\end{eqnarray}

Let us now establish the ODEs obeyed by the three quantities $\Rone$, $\Sone$ and $\Rone \pm \Sone$,
which, according to (\ref{def-D-calD}),
respectively determine the values of the three determinants in Table \ref{table:identity}. It turns out that the equation for $R\pm S$  is independent of the sign chosen: the distinction is hidden in the initial conditions.
These three equations, especially for the last two of them, are more compactly displayed using three intermediate functions using the temporary notations
\begin{eqnarray}
& & {\hskip -10.0 truemm}
\frac{\RHO}{\Rone}=\frac{\Ss}{\Sone}=\frac{\Ts}{\Rone\pm\Sone}=\sinh(2 \nu T),
\label{eqode1RSNotation}
\end{eqnarray}
so that in particular  (\ref{eqode1RS}) acquires the very simple form (where ${}'=\dd /\dd T$)
\begin{eqnarray}
& & {\hskip -10.0 truemm}
{\rho}'^2 - {\Ss'}^2 -(2 \nu \rho)^2 + (2 \nu \THETAZ \Ss)^2=0.
\end{eqnarray}
We remark that 
the common sinh factor in the definitions above is essentially the modulus of the 
Hopf factor in geometry, cf. Appendix \ref{AppendixGeometry}. We have also to emphasize that these notations are used only in this subsection, and that this $\rho=\rho(T)$ bears no relation to the inclusive $n$-point correlation functions $\rho_n$ defined in Eq.~(\ref{defrhon}) of the previous subsection, or  that  $\sigma=\sigma(T)$ is distinct from the so-called Okamoto--Jimbo-Miwa sigma-form Eq.~(\ref{sigmaOJMh}) satisfied by the reduced $\PVI$ Hamiltonian.

The second order ODE for  $\Rone=\Rone(T)$, the resolvent kernel function at coincident points,
is still manageable without the sinh factor: \begin{eqnarray}
& & {\hskip -10.0 truemm}
(\Rone''+4 \nu \coth(2 \nu T) \Rone')^2
\nonumber\\ & & {\hskip -10.0 truemm}
-4 \Rone' \left\lbrack 2 {\Rone'}^2 + 8 \nu \coth(2 \nu T) \Rone \Rone' + 4 \nu^2 \theta^2 \Rone' 
  - 8 \nu^2(1-\coth^2(2 \nu T)) \Rone^2 \right\rbrack =0,
\label{eqode2R}
\end{eqnarray}
One of the  advantage to display it this way is 
that up to a simple rescaling it coincides direcly with the equation
for the mean curvature of Bonnet surfaces.
Another one is that
 for  $\theta=N$ and $\nu=\ii$, where the $\coth$ functions above become ordinary trigonometric ones,
 it is identical to \cite[Eq.~(5.70)]{TWFred1994}, which provides a useful check of the algebra. 

The ODE for $\mu(T)$ is independent of the sign in $\Rone\pm\Sone$:
\ba
& &
\lp \mu'' - 2 \nu \coth{(2 \nu T)\mu'}\rp^2 
- \frac{4 \mu'^3}{\sinh{( 2 \nu T)}}
+  \lp 4 \nu^2(1-\vta^2 -\coth^2{(2 \nu T)}) +8 \frac{2 \nu \coth{(2 \nu T)}}{\sinh{(2 \nu T)}}\mu \rp \mu'^2 
\nonumber\\ & &
+\frac{8 \nu^2}{\sinh{(2 \nu T)}}\lp 2 \nu \vta^2 \cosh{(2 \nu T)} -2 \mu \rp \mu \mu'-16 \nu^4 \vta^2 \mu^2=0.
\label{eqode2T}
\ea

Finally, the second order ODE for $\Sone(T)$ is
\begin{eqnarray}
& & 
(\Ss'' + 8 \Ss^3 - (2 \nu)^2 \THETA^2 \Ss)^2
-  16 \coth^2 (2 \nu T) \Ss^2
\left\lbrack {\Ss'}^2 + 4 \Ss^4 - (2 \nu)^2 \THETA^2 \Ss^2 \right\rbrack =0.
\label{eqode2S}
\end{eqnarray}

As to the thinning parameter $\xi$, evidently absent of these three ODEs,
it is determined by the Neumann expansion of the resolvent.


\subsection{Integration in terms of a system of intertwined $\PVI$ functions}
\label{section-integration}

The integration of these second order second degree ODEs is straightforward.
The method is to first look at the movable poles of $\Rone$, $\Sone$ and 
$\Rone \pm \Sone$,
then at the second order ODE defined by the first square of the Gauss decomposition
of their ODE.
Given these two essential pieces of information,
one selects in the existing classifications of second order second degree ODEs
which have the Painlev\'e property
\cite{ChazyThese,BureauMIII,CosgroveScoufis1993,CosgroveChazy2006}
the very few possible matches.
Finally, one tries to identify the ODE at hand and the selected match
by a homographic transformation of the dependent variable,
see details in \cite{CMBook2}.

The ODE for $\Sone$ has only four movable simple poles 
with residues $\pm 1/2$ and $\pm (\ii/2) / \sinh((2 \nu) T)$ 
and its first square defines an elliptic ODE with two simple poles,
therefore it is likely to be equivalent to a particular equation $\CVI$ of Chazy
\cite[p 342]{ChazyThese},
whose general solution $w(x)$ is recalled in Appendix \ref{Appendix-PVI-HVI-CVI}.

The ODE for $\Rone$ (resp.~$\Rone \pm \Sone$) has only one movable simple pole of residue $-1/2$ (resp.~$-1$) 
and its first square defines a linear ODE,
therefore it is likely to be equivalent to a particular so-called sigma-form 
of $\PVI$
\cite[p 340]{ChazyThese},
(\ref{sigmaOJMh})
whose general solution $h(s)$,
recalled in Appendix \ref{AppendixOkamoto},
is defined so that $h(s)/(s(s-1))$ is a tau-function of $\PVI$ 
(one movable simple pole of residue unity).

Both guesses are true, and the explicit integration is the following.
We start with the most fundamental object, namely $\Sone$,
because,
although it does not describe a probability law,
it is essentially identical to the most intrinsic object in geometry,
namely the metric function $\calG$.
Moreover, given $\Sone$, the probability laws $\Rone$ and $\Rone \pm \Sone$
follow by performing one quadrature of the Gaudin relation (\ref{eqodeGaudin}), cf.
(\ref{resStoD2det}).

Up to scaling factors,
the equations (\ref{eqode2R}), (\ref{eqode2S}) obeyed by the resolvent functions $\Rone(T)$ and $\Sone(T)$
are identical to  (\ref{Hazzi}), (\ref{EG2}), respectively satisfied by the mean curvature 
and the metric function of Bonnet surfaces, see Appendix \ref{AppendixGeometry}.


Up to rescaling, the ODE (\ref{eqode2S}) is identical to 
\begin{eqnarray}         
& & {\hskip -10.0 truemm}
(w'' -2 w^3 - D_2 (2 \nu)^2 w - D_3 (2 \nu)^3)^2
+ \left\lbrack 2 \coth (2 \nu T) \left(w-\frac{D_1 (2 \nu)}{\cosh(2 \nu T)}\right)\right\rbrack^2
\nonumber\\ & & {\hskip -10.0 truemm}
\times
\left\lbrack {w'}^2 - w^4 - D_2 (2 \nu)^2 w^2 -2 D_3 (2 \nu)^3 w - D_4 (2 \nu)^4\right\rbrack =0,\
\nonumber\\ & & 
D_1=D_3=D_4=0, \quad
\Sone=\frac{2 i w}{\sinh(2 \nu T)},
\label{eqCVI} 
\end{eqnarray}
an equation first isolated by Chazy \cite[Eq.~(C,V) p.~342]{ChazyThese}
(following the convention of Chazy, the factor $i$ ensures
that $w(T)$ admits two simples poles of residues $\pm 1$).

The monodromy exponents of the associated $\PVI$ take two sets of values
(with $\theta_\alpha=\theta_\beta$)
(see Appendix \ref{AppendixOkamoto}),
\begin{eqnarray}
& &
\left\lbrace
\begin{array}{ll}
\displaystyle{
\left(\frac{1}{4},\frac{1}{4},\THETA^2,0\right) 
\hbox{ or } 
\left(\frac{(\THETA-1)^2}{4},\
      \frac{(\THETA-1)^2}{4},\
			\frac{ \THETA^2}{4},\
			\frac{ \THETA^2}{4}\right)
}\\ \displaystyle{
=
            (\theta_\alpha^2,\theta_\beta     ^2,\theta_\gamma     ^2,\theta_\delta^2)
\hbox{ or } (\theta_\alpha^2,\theta_\beta     ^2,\theta_\delta     ^2,\theta_\gamma^2).
}
\end{array}
\right.
\label{eqmonoS}
\end{eqnarray}

For $\Rone$, the map between (\ref{eqode2R}) and (\ref{sigmaOJMh}) is
\begin{eqnarray}
& &
\Rone(T)=2 \nu h_{\Rone}(\Rx),\
\Rx= \frac{1+ \coth(2 \nu T)}{2},\
\label{eqRtoy}
\end{eqnarray}
and the monodromy exponents of the corresponding $\PVI$ 
are anyone of the two equivalent sets
\begin{eqnarray}
& &
(\theta_\alpha^2,\theta_\beta     ^2,\theta_\gamma     ^2,(\theta_\delta-1)^2)=
(0,\THETA^2,\THETA^2,0) \hbox{ or }
(\THETA^2,0,0,\THETA^2).
\label{eqmonoR}
\end{eqnarray}

Both $\Rone$ and $\Rone \pm \Sone$ are Hamiltonians evaluated on the equations of motion  of two $\PVI$  
with different monodromy quadruplets and distinct independent variables,

\begin{eqnarray}
& &
\left\lbrace
\begin{array}{ll}
\displaystyle{\HVI(p(s),q(s),s) \D s=2 \Rone \D T 
=\frac{h_{\Rone} \ \D s}{s(1-s)},\
}\\ \displaystyle{
s \hbox{ or }1-s=\frac{1+ \coth(2 \nu T)}{2},\
}
\end{array}
\right.
\label{eqMds}
\end{eqnarray}
and
\begin{eqnarray}
& &
\left\lbrace
\begin{array}{ll}
\displaystyle{\HVI(P(t),Q(t),t) \D t=(\Rone \pm \Sone) \D T
=\frac{\D t}{t(1-t)} \left(h_{\pm} +\frac{t}{16}-\frac{\THETA^2}{8}\right),
}\\ \displaystyle{
t \hbox{ or } \frac{1}{t}=\coth^2(\nu T),
}
\end{array}
\right.
\label{eqMdt}
\end{eqnarray}
in which the subscripts remind that the two $h$'s are different.

The ODE for $h_{\pm}(t)$ is the particular case 
$\vabcd=(1/2)(\vartheta,\vartheta,1-\vartheta,1-\vartheta)$,
see Eq.~(\ref{sigmaOJMh}) in Appendix \ref{Appendix-PVI-HVI-CVI}.

For reference, the precise correspondence between the two resolvent functions
and the geometric quantities is 
\ba
-\calH(x)= 2 R(x,x;\xi) &=& 2 \nu h(s(x)), \quad s(x) = \frac{1+\coth{( \nu x)}}{2}\\
\calG(x) = \frac{\sinh{( \nu x)}}{\nu} 2 R(-x,x;\xi) &=& G(t(x)), \quad t(x) = \tanh^2{(\nu x/2)} 
\ea

\textit{Remark}.
 \label{CorollaryGaussBonnet}
The persistence exponent $\kappa(m)$ is \textit{also} characterized by an intrinsic geometric invariant,
namely the Willmore energy of the corresponding  Bonnet surface,
defined as the integral of $\calH^2 -\calK$
over the whole surface,
 \be
\label{Will0}
\calE(m) \coloneqq \int_{{\Sigma}_>} \! \dd \calA
\lp \calH^2 -\calK \rp 
=  2\pi \,  \calH(\RE{z};1-m^2) \Big{|}^{\RE{z}=0^+}_{\RE{z}=+ \infty} =
2 \pi \, \kappa(m) - \frac{1-m^2}{4}.
\ee

\textit{Remark}.
The invariance of (\ref{eqCVI}) under parity only requires $D_1=D_3=0$.
The consideration of the metric of a Bonnet surface in the Riemannian manifold
$\Rcubec$ \cite[page 120]{BEBook} instead of the flat Euclidean space $\Rcube$
allows both $D_2$ and $D_4$ to be nonzero,
the coefficient $D_4$ being proportional to the extrinsic curvature.
\medskip

\subsection{The distinguished Painlev\'e VI governing persistence}
\label{section-Folding-symmetries}

The kernel of persistence corresponds to $\theta=1/2$.
In this case (and only in this case),
one of the admissible choices for the
signed quadruplet $\la \theta_\alpha,\theta_\beta ,\theta_\gamma,\theta_\delta\ra$
in (\ref{eqmonoS}) is
\begin{eqnarray}
& & \la \theta_\alpha,\theta_\beta ,\theta_\gamma,\theta_\delta\ra=\la\frac{1}{4},\frac{1}{4},\frac{1}{4},\frac{1}{4}\ra.
\end{eqnarray}
(We recall that signs of the monodromy exponents do matter in the Hamiltonian formulation of $\PVI$,
in particular for the associated impulsion,
see Appendix \ref{Appendix-PVI-HVI-CVI}.)
Then there exists an algebraic transformation leaving $\PVI$ form-invariant
(i.e.~only changing its four parameters), due to Kitaev \cite{Kitaev1991},
recalled as Proposition \ref{PropTOS2} in Appendix \ref{Appendix-PVI-HVI-CVI},
whose two successive applications map this quadruplet to $\la 0,0,0,1 \ra$,
\be
\label{rest421}
\la\frac{1}{4}, \frac{1}{4},\frac{1}{4},\frac{1}{4} \ra \hookrightarrow 
 \la \frac{1}{2},0,0,\frac{1}{2} \ra                      \hookrightarrow \la 0,0,0,1 \ra.
\ee

This folding transformation has later been shown by Manin \cite{Manin}
to be identical to a Landen transformation between elliptic functions
in the elliptic representation of $\PVI$.

The above quadruplet (\ref{rest421}) displays two advantages,
due to the cancellation
of the last three terms of the impulsion $p$ Eq.~(\ref{ps2qs}).
Introducing the Lagrangian
associated to the Hamiltonian (\ref{defHpqs}),
\be
\label{def-Lagrangian}
L_\mathrm{VI}\lp \dot{q},q,t \rp \coloneqq p \ \dot{q} - H_\mathrm{VI}(p(t),q(t),t).
\ee
the first advantage is a multiplication by a factor four
of the two respective Lagrangians, 
a property symbolically written as
\be
4 L_{\la 4 \times 1/4\ra} (\dot{q},q,t) = 
L_{\la 3 \times 0, 1 \ra} (\dot{Q},Q,t)
\ee
where the $\PVI$ function $Q=Q(t)$ appearing on the right-hand-side above has therefore Manin's $\PVI$ coefficients $\abcd=\lb 0,0,0,0 \rb$.
The second advantage is the exceptional equality of the Lagrangian and the Hamiltonian
for the quadruplet of monodromy exponents $\la 0,0,0,1 \ra$.

This  justifies what we have announced much earlier in our Introduction and in our Theorem \ref{thm:B}: the occurrence
of  Manin's $\PVI$ with these distinguished coefficients  for the persistence problem.  Notice finally that
since one of the admissible set of the monodromy exponents for the mean curvature of a $\Kta$-Bonnet surface can be taken as (cf. Proposition \ref{PropBonnetH} in Appendix \ref{AppendixGeometry}) 
$$
\la \theta_\alpha,\theta_\beta ,\theta_\gamma,\theta_\delta\ra
= \la \ta, 0,0,1-\theta\ra,
$$
this  Manin $\PVI$ is \textit{also} a Bonnet $\PVI$: the one with $\ta=0$ determined by the kernel $K_0$, Eq.~(\ref{defK0}), whose Fourier transform by (\ref{resFTKta}) is itself essentially the square of the (self-dual) sech-kernel $\Ksech$!




\subsection{Transcendental nature of the persistence distribution}
\label{section-transcendence}

Due to the presence of these very special symmetries, one could wonder if
the Bonnet-Manin $\PVI$ controlling  the persistence distribution function is
a genuinely transcendental  one, or a (particularly) convoluted classical one, namely whether one could express it  through some first order or linearizable ODE, 
or as an algebraic function.

Indeed, since there exists a choice of signs of the four monodromy exponents 
making their sum unity,
the solution of the associated $\PVI$ could be the logarithmic derivative of a hypergeometric function.
It could also be one of the 48 algebraic solutions of $\PVI$ \cite{LT2014}.

We shall give 
two different justifications to show this is not the case, thereby establishing the transcendental nature of the persistence probability density. 
The crux of the matter for both arguments is
that $\Sone$, the second logarithmic of the Fredholm determinant for the sech kernel, cannot vanish.

\begin{proof}[First proof].
According to its link (\ref{solgtgen}) between $\Sone$ and $\CVI$,
the impulsion $P$ of the Hamiltonian $\HVI$ Eq.~(\ref{defHpqs}) is nonzero,
therefore the hypergeometric possibility (which is defined by $P=0$) is ruled out.

There remain to discard possible algebraic solutions.
According to the list of \cite{LT2014}, the only algebraic solution matching
the monodromy exponents $(1/4,1/4,1/4,1/4)$ is $Q=\sqrt{t}$.
This would again imply a zero impulsion and therefore $\Sone=0$, yielding a contradiction.
\end{proof}

\medskip

\begin{proof}[Second proof].
The second proof relies on a result of Borodin and Okounkov (Appendix \ref{AppendixBOconnection}).
Given any hyperbolic kernel $\Kta$ in the class (\ref{defKtanu}), and thus in particular for the sech kernel
there exist scalars $A,B,C$ such that
\be
\forall \xi, 0 \le \xi <1\
\forall T: \
\log \Det(\Id - \xi \Ksech)\upharpoonright_{[-T,T]}= -2\int_0^T \!\dd x R(x) = A \, T + B + C(T),
\ee
in which $C(T)$ is the Fredholm determinant of a kernel 
which vanishes exponentially fast as $T \to + \infty$.
Taking the second derivative w.r.t.~$T$, one obtains
\be
\forall T: \ 
\frac{\D^2}{\D T^2}\log \Det(\Id - \xi \Ksech)
= 2 R'(T)=-4S^2(T)
=\frac{\D^2}{\D T^2} C(T),
\ee
and this second derivative of $C(T)$ and $S(T)$ are never equal to zero,
cf.~(\ref{BO-C-trans}).
\end{proof}

\section{Conclusion}
\label{SectionConclusion}

In this work, we have shown that the full persistence probability distribution
for the one-dimensional Ising--Potts model in the stationary scaling regime
is governed by a distinguished Painlev\'e~VI system, which we have termed
the \emph{Bonnet--Manin Painlev\'e~VI}.

{}From a technical standpoint, this result places the persistence problem within
the general framework of integrable Fredholm Pfaffians.
In this respect, the \emph{sech} kernel  determines the Bonnet--Manin Painlev\'e~VI in a way which is structurally similar to the way that  the sine
kernel determines  the Gaudin--Mehta Painlev\'e~V distribution  in random matrix theory.

At a more conceptual level, our results suggest that the persistence distribution
occupies, for nonequilibrium coarsening dynamics, a role analogous to that of
Tracy--Widom distributions in equilibrium and near-equilibrium statistical systems:
a universal law characterized by an integrable kernel, a Painlev\'e equation,
and a nontrivial connection problem.

We expect that this perspective will be useful in other persistence and
first-passage problems where Pfaffian structures and integrable kernels
naturally arise.

Finally, we note that Fredholm Pfaffians and Painlev\'e equations also play
a central role in the description of universal fluctuations in growth processes
belonging to the Kardar--Parisi--Zhang (KPZ) universality class.
In particular, one of the Tracy--Widom distributions (the $F_1$ one), which governs height
fluctuations for curved initial conditions in several KPZ settings, admits
a Fredholm Pfaffian representation involving the Airy kernel and is associated
with a Painlev\'e~II equation.

A remarkable feature of KPZ dynamics is the persistence of memory of the
initial condition through a single curvature parameter, which controls both
the limiting distribution and its scaling properties.
In the present persistence problem, the exponent $\kappa(m)$ plays an
analogous role, encoding the dependence on the initial magnetization through
the single parameter $\xi = 1 - m^2$.

From a structural viewpoint, it is therefore natural to ask whether these
parallels reflect a deeper connection.
Since the Painlev\'e~VI equation governing persistence is the most general
member of the Painlev\'e equations, and since lower Painlev\'e equations such
as Painlev\'e~II arise in KPZ through well-known confluence limits, one may
speculate that suitable scaling or degeneration limits could relate
persistence distributions to KPZ fluctuation laws.
Exploring such connections lies beyond the scope of the present work, but we
believe that the framework developed here provides a natural starting point
for addressing this question.

\medskip

\noindent\textbf{Statements and declarations}

%
 The authors have no relevant financial or non-financial interests to disclose.
\\ The authors have no competing interests to declare that are relevant to the content of this article.
\\ All authors certify that they have no affiliations with or involvement in any organization or entity with any financial interest or non-financial interest in the subject matter or materials discussed in this manuscript.
\\The authors have no financial or proprietary interests in any material discussed in this article.

\medskip
\noindent \textbf{Data availability statements} 



The authors declare that the data supporting the findings of this study are available within the paper.


\section*{Acknowledgements}

During the many years that lasted this persistence endeavor, the first author has benefited from  useful remarks, comments, or encouragement  by
many 
colleagues, in particular
A.N.~Borodin, C.~Colleras, B.~Derrida, P.~Di~Francesco,  G.~Korchemsky, O.~Lisovyy, J.-M.~Maillard,
K.T.R.~McLaughlin, V.~Pasquier and G.~Schehr.

Both authors thank the referees and the handling editor for their careful reading and constructive suggestions, which have helped improve the presentation of the manuscript.
\appendix

\setcounter{equation}{0}\renewcommand\theequation{A\arabic{equation}}
\section{The DHP formula for the  persistence probability as a Pfaffian gap-spacing probability with the sech kernel}
\label{AppendixDHPqformula}

One of the crucial points of the present paper is to recognize that
the formula \cite[Eq.~(29) p. 773]{DHP1996} 
for the persistence probability of Potts spins in the scaling r\'egime
can be recast as the following one for $\pm$ Ising spins \cite[Eq.~(8)]{Dornic2018}),
\be
 \label{beautyS2}
P_0^+(\ell;m) 
\coloneqq \frac{1}{q} \ p_0([t_1,t_2];m)\Big{|}_{1/q= \frac{1 + m}{2}} \ 
=\frac{1 + m}{2}\frac{\calD^+ +\calD^-}{2} + \frac{1 - m}{2} \ \frac{\calD^+ -\calD^-}{2},
\ee
because in the random initial conditions
one can identify a given colored Potts spin 
as a $+$ spin with  probability  $1/q=(1+m)/2$. 
By  reversing the signs of all  spins in the random initial condition (which does not affect the subsequent dynamics) one also realizes that it holds  that
\be
P_0^+(\ell;-m) \equiv P_0^-(\ell;-m),
\ee
where the $P_0^\pm(\ell;m)$ have the same meaning as in (\ref{thm:A}). As for our $p_0([t_1,t_2];q)$, it 
is  identical to the expression $R(q;t_1,t_2)$ of \cite[Eq.~(29)]{DHP1996}.

In the notation defined in \eqref{def-calD},
$\calD^\pm = \calD^\pm(\ell;1-m^2)$
are the Fredholm determinants of, respectively, the even and odd parts of the 
sech kernel,
with the correspondence for the independent variable
\be
1 \ll t_1 \ll t_2 \quad \mathrm{with} 
\quad \ell =\log\lp {t_2}/{t_1}\rp  >0,
\ee

The proof of the identity (\ref{beautyS2}) was sketched in the unpublished work \cite{Dornic2018}  by the first author. It just relies on applying the by now standard Tracy-Widom
technique \cite{TWGOSE1996} developed to study the orthogonal and symplectic ensembles
of random matrix theory,
which recasts a matrix Pfaffian Fredholm determinant in terms of
two scalar functions linked
by the Gaudin relation (\ref{eqodeGaudin}). The derivation
of (\ref{beautyS2}) is in fact valid for any even difference kernel acting on a symmetric interval. In other words, it is a signature of the intrinsic nature of the Pfaffian persistence probability.

As to the existence of  scaling limit in the particular case of the Ising model with $m$-magnetized random initial conditions, it has been proven since, 
in the rigorous probabilistic work \cite{FTZ2022}.

All in all, this justifies all the statements made in our Theorem (\ref{thm:A}), but for the ODE (\ref{ODEHThmA}) satisfied by the resolvent for the sech kernel, which is established in Section \ref{section-closed-differential-system}, and the value (\ref{reskappam})
of the persistence exponent, which comes from the lowest-order term of the Borodin-Okounkov formula of  Appendix \ref{AppendixBOconnection}
for the geometric means of the Fredholm determinant $\calD=\calD^+.\calD^-$

\section{The Painlev\'e VI solutions for  the mean curvature and the metric of Bonnet surfaces in $\RR^3$}
\setcounter{equation}{0}\renewcommand\theequation{B\arabic{equation}}
\label{AppendixGeometry}

In this  Appendix, we recall how  the  constitutive Gauss-Codazzi equations,  which determine locally an ordinary surface immersed in the usual 
three-dimensional space, 
can be  \textit{globally} solved in terms of $\PVI$ transcendents for a   particular class of isometry-preserving  surfaces
 introduced and studied by the French geometer Bonnet in 1867, and  bearing his name since then.

It was
Bobenko and Eitner who found in 1994 that the mean curvature  for every Bonnet
surface is in fact a certain $\PVI$ tau-function. Here we
present the $\PVI$ solution for their metric.  Remarkably, we find that not only the latter  satisfies a rather esoteric 
sibling of the $\PVI$ equation that Chazy mentioned --- without displaying its solution! --- as  early as 1911 \cite{ChazyThese}, but we also show that it
is  related  to the    
$\PVI$ determining
 the mean curvature function by a certain folding  transformation,  of algebraic  and non-birational nature.

 Of course, the extraordinary coincidence is
that the ODEs satisfied by the determinants giving the persistence probability
are but a particular case of the  ones occurring in  the  Bonnet reduction of the Gauss-Codazzi PDEs. 

This Appendix is organized  as follows. We start by reviewing  in a pedagogical way some  classical material concerning  the differential description  of the geometry of ordinary (bi-dimensional) surfaces in   (Euclidean space) $\RR^3$.  Along the way  
 we review the 
 nomenclature and we set the appropriate notations 
 needed for 
  our purposes.
  The $\PVI$ solution for the mean curvature function of Bonnet surfaces,  based on  \cite{Bobenko1994, BEBook, RC2017}, is recalled as Proposition \ref{PropBonnetH}. The main novelty
  of this Appendix, the
  $\PVI$ solution for the metric and its relation to the former one, is given as 
  Proposition \ref{PropBonnetG}.

 Consider in the usual
  three-dimensional Euclidean space a  smooth enough 
   surface $\Sigma$, thus
   determined --- after Gauss ---  by two fundamental real quadratic forms  $I, II$ defined at each of its points $\mathbf{r}\in \RR^3$. 
   We take for granted
   the existence of  a chart of complex conformal coordinates 
   $(z, \zb) \equiv (x+ \ii y,  x-\ii y)$ 
  which  parametrize (at least locally) the domain $\Sigma$ 
  of $\RR^2$ of the corresponding surface \textit{immersion}  $\mathbf{r} = \mathbf{r}(z,\zb)$ in $\RR^3$, 
   so that these two 
   quadratic forms  can be equivalently expressed as 
  \ba
  \label{defI}
I &\coloneqq & \dd \mathbf{r} \cdot \dd \mathbf{r} = \sqrt{\det{[g_{ij}]}} \lp \dd x^2 + \dd y^2 \rp = 
\calG^{-2}\dd z \dd \zb,\\ 
 \label{defII}
II & \coloneqq &  -\dd \mathbf{r} \cdot \dd \mathbf{n} =
 \frac{1}{2}{\calJ} \dd z^2 +
\calH \dd z \dd \overline{z} + \frac{1}{2}\overline{\calJ} \dd \overline{z}^2, 
\ea
where  the $\cdot$ in the definitions above represents   the usual (Euclidean)
scalar product in $\RR^3$.

In the first fundamental form, and before recalling what the (real-valued) function   $\calH=\calH(z,\zb) $ and the (complex-valued) one $\calJ=\calJ(z,\zb)$ functions stand for in the second, we have used  both $(\dd x, \dd y)$ and $(\dd z, \dd \zb)$ viewpoints,
assuming for the  latter that the conformal coordinates  are  also  \textit{isothermal}. 
This means that the  positive definite 
 $2\times2$ matrix $\lb g_{i j} \rb$ giving the metric
is diagonal. Conveniently for what follows, we express its sole non-zero and real-valued coefficient as  $g_{i j} \equiv \delta_{i j}/\calG^{2} $, 
with $\calG^2=\calG^2(z,\zb)$.
It also  turns out that the existence of such isothermal 
coordinates is always ensured for Bonnet surfaces,
except at their single \cite{Roussos1999} ``umbilic'' point
(where the two principal curvatures coincide --- more on this around (\ref{diffkappas}) below).

The function 
$\calG$, which in the two fundamental quadratic forms 
 only appears through the square of its inverse, 
is  the \textit{conformal factor} for the metric. 
Indeed, when one makes a
re\-pa\-ra\-me\-tri\-za\-tion  $z \hookrightarrow w$, $z=f(w)$  of the surface chart, with $f$ analytic and $f'\neq 0$ 
within the domain $f^{-1}(\Sigma)$, the piece $1/\calG^{2}$ transforms as  $|f'(w)|^2/ \calG^{2}$.  
Bearing in mind these relationships, we shall indifferently refer to  either $1/\calG^2$, $\calG^2$, or $\calG$, as the conformal factor for the metric,
the metric factor, or even simply  the metric.

As for the second fundamental  form (\ref{defII}) (not necessarily a positive definite one), it involves the oriented  unit vector $\mathbf{n}$  normal to the  tangent plane  spanned by
$(\partial \mathbf{r},\overline{\partial} \mathbf{r})$, so that equivalently  $II = \dd^{2}\mathbf{r}\cdot \mathbf{n}$. Here and elsewhere, 
partial derivatives with respect to the conformal coordinates 
are defined and abbreviated  as $\partial \equiv \partial_z \coloneqq \frac{1}{2}\lp \partial_x- \ii \partial_y\rp$, $\overline{\partial} \equiv \partial_{\zb} \coloneqq \frac{1}{2}\lp \partial_x + \ii \partial_y\rp$.

The necessary conditions that  
the three coefficients $(\calG,\calH,\calJ)$ which  appear in the two fundamental quadratic forms 
have to satisfy (as   functions of $z,\zb$)  to \textit{locally} describe a surface  date back to  the nineteenth century, when they were first written down by  Gauss and Codazzi (along with and independently by Peterson and Mainardi  for the second ones \cite{Phillips}), to wit
\ba
\label{defGauss}
4  \partial \op (\log{\calG}) &=&  \calG^{-2} \calH^2  -  \calG^2|\calJ|^2 \\
\label{defCodazzis}
\calG \overline{\partial}\calJ &=& \calG^{-1} {\partial}\calH \quad \mathrm{and} \quad
\calG {\partial}\overline{\calJ} = \calG^{-1} \overline{\partial}\calH, \quad
\mathrm{its  \, complex  \, conjugate}.
\ea

The  other real-valued  function $\calH=\calH(z,\zb)$ that appears in the above under-determined system of non-linearly coupled PDEs   is the  \textit{mean curvature} 
$\calH \coloneqq \frac{\kappa_1+\kappa_2}{2}$, that is to say 
the  arithmetic mean of its 
two \textit{principal curvatures} $\kappa_1, \kappa_2$.
The latter  are defined at each  point
of the surface as  the two 
eigenvalues --- assumed  non-degenerate ---  of the 
so-called Weingarten's 
``shape operator''. This  linear operator (living abstractly in the tangent space to the surface) can be  represented by a $2\times 2$ real symmetric    matrix, simply     built  
 from
those associated to 
the  two  fundamental  quadratic  forms in the $(\dd x, \dd y)$ basis
\be
\label{defshape}
 [II] \times
 [I]^{-1}= \begin{bmatrix} {\calH}  +\calG^2 \, \mathrm{Re}\calJ & -\calG^2 \,  \mathrm{Im}\calJ\\ -\calG^2 \, \mathrm{Im}\calJ & {\calH}  - \calG^2 \, \mathrm{Re}\calJ  \end{bmatrix}. 
 \ee

By contruction, the trace   of this  
shape operator is   $2 \calH=\kappa_1+\kappa_2$, twice the mean curvature, while
  its determinant 
 defines Gauss --- or \textit{total} --- curvature
 \be
 \label{def1K}
 \calK \coloneqq \kappa_1 \kappa_2
 ={\det{[II]}}/{\det{[I]}}=\calH^2-\calG^4 |\calJ|^2.
 \ee
  That  quantity is not only invariant under conformal reparametrization, but   it is also \textit{intrinsic}.
 Indeed, from the first of the Gauss-Codazzi equations, (\ref{defGauss}),
 one also reads that
 \be \calK = 4 \calG^2 \partial \op (\log{\calG}), 
 \ee 
  where the right-hand-side involves
   the Laplace-Beltrami operator  for the induced 
   Riemannian manifold  (here two-dimensional). That 
 the total curvature  depends  in fact  only on the metric of the surface without the need for any $\RR^3$-embedding is the gist of Gauss celebrated \textit{Theorema egregium} (``remarkable theorem'').

 The third coefficient $\calJ=\calJ(z,\zb)$ which appears in the second fundamental form and in the Gauss-Codazzi equations  is defined through a  conformally-invariant $(2,0)$ two-form  $\calJ \dd z^2 \coloneqq \lp \partial^2 \mathbf{r} \cdot \mathbf{n}\rp \dd z^2 $ called the quadratic Hopf differential. Henceforth we shall simply refer to 
$\calJ=\calJ(z,\zb)$ as the \textit{Hopf factor}. Contrarily to $\calG$ and $\calH$, it is    
 generically a  complex-valued object, its modulus $|\calJ|$ being always nonzero away from the so-called 
\textit{umbilic points} of the surface, where (by definition)   the two principal curvatures coincide. Indeed, by using the expression 
(\ref{def1K}) for the determinant of the shape operator and combining it  with the definitions for the mean and Gaussian curvatures,   a string of (trivial) identities  produces  eventually the (square of)
 the so-called
\textit{skew curvature} $\calL$, that is to say (half) the difference between the two principal ones 
\be
\label{diffkappas}
\calG^{4}|\calJ|^2=\calH^2 -\calK  = \lp \frac{\kappa_1+ \kappa_2}{2}\rp^2- \kappa_1 \kappa_2 =\lp \frac{\kappa_1 - \kappa_2}{2}\rp^2\equiv    \calL^2, \quad \calL \coloneqq \frac{\kappa_1 -\kappa_2}{2}.
\ee

Before addressing the historical Bonnet problem for which the complex-valued nature of the Hopf factor  plays a key r\^ole, let us introduce for a generic surface $\Sigma$ another natural  quantity  that one can define using (\ref{diffkappas}),  viz. its  \textit{Willmore energy}  $\calE=\calE[\Sigma]$. 
Its consideration  will be very helpful later on to relate 
  our Fredholm determinants solutions of the Gauss-Codazzi equations to \textit{global} properties of   Bonnet surfaces. 
  
  To wit, observe that by
integrating (\ref{diffkappas})
over the whole surface $\Sigma$ (or part of it) with the two-form area 
\be 
\label{def-calA} \dd \calA \coloneqq  \sqrt{\det{[g_{ij}]}} \, \dd x\wedge \dd y \equiv  \calG^{-2} \dd \!\RE{z} \wedge \dd\! \IM{z},
\ee
one obtains by construction  a non-negative and conformally invariant global measure of the surface asphericity
\be
\label{Will1}
\calE \coloneqq \int_\Sigma \! \dd \calA   \lp \calH^2 -\calK \rp =\int_\Sigma \, \lp\dd \!\RE{z} \wedge \dd\! \IM{z}\rp \,  \calG^2|\calJ|^2,
\ee
 since by (\ref{diffkappas})  the integrand is locally non-zero  if and only if  $\kappa_1 \neq \kappa_2$. In fact, this Willmore energy is  a fundamental quantity, which appears time and again and in one guise or in another in various domains of the natural sciences. We briefly mention a couple of examples below.

Let us first remark that   (\ref{Will1}) involves \textit{Euler's characteristic}, defined  in the continuous context here by Gauss integral curvature, viz. $\calX_{\mathrm{E}} \coloneqq (2 \pi)^{-1}\int_\Sigma \! \dd \calA \,  \calK$. 
For {compact} surfaces of fixed genus,   this constant topological term is usually discarded to deal with
the shifted (and still non-negative)
  Willmore
functional
\be
\label{Will2}
\widetilde{\calE} \coloneqq \int_\Sigma \! \dd \calA \, \calH^2 \equiv \calE +2 \pi \calX_{\mathrm{E}},
\ee
its one-dimensional reduction $\widetilde{\calE} \hookrightarrow \int \! \dd l \, \kappa^2$ ($l=$ arc length) being  nothing but Euler-Bernouilli's  \textit{elastica} (see, e.g., \cite{MatsutaniElastica,Mumford}, and references therein)!

As for this famous problem, a long-standing issue for geometers   and which is still of interest nowadays  \cite{Taimanov2006,Taimanov2022} has been to globally minimize that ``conformal area" 
$\widetilde{\calE}$  by varying $\calH$.
It turns out that this very same quantity   also shows  up in models of biological membranes, where $\widetilde{\calE}$ represents the so-called (Canham-)Helfrich free 
energy \cite{Canham1970,Helfrich1973} and even in string theory, where Polyakov \cite{Polyakov1981} considered it as the  extrinsic action for two-dimensional quantum gravity.

This being said, 
a classical theorem by Bonnet  
ensures that \textit{any} solution $(\calG,\calH,\calJ)$ to  the Gauss-Codazzi equations~(\ref{defGauss})--(\ref{defCodazzis}) determines a \textit{unique} 
surface in $\RR^3$ (this of course locally, and up to ``rigid motion'',  i.e.  up to arbitrary finite translations and rotations). For instance, the so-called constant mean curvature surfaces are obtained with $\calH=\mathrm{const}$, and their metric obeys the Liouville PDE (if this $\mathrm{const}=0$ as for minimal surfaces), or the  (by now!) classically integrable sin(h)-Gordon PDE (if that $\mathrm{const} \neq 0$).

Conversely, one can wonder whether there is any redundancy in  the Gauss-Codazzi equations, that is to say which data among the three functions $(\calG,\calH,\calJ)$  can be dispensed of in the nonlinear system of under-determined   PDEs they obey, and what family of surfaces  in $\RR^3$ would thereby be  defined. 

Following this geometric line of thought, Bonnet  raised and answered
in  his 1867's seminal  work
     the problem that still bears his name: 
     Given a real surface in $\RR^3$,
how to determine all surfaces which are \textit{applicable} on it, that is to say which
possess (up to conformal transformations)  the same two principal  curvatures $\kappa_{1}, \kappa_{2}$ ?  
  Since  isometries of $\RR^3$  conserve  the metric hence in particular the Gaussian curvature $\kappa_1 \kappa_2$, these  considerations led Bonnet  
to the discovery and the  characterization of   a family of \textit{isometry-preserving} surfaces, all having  the same  --- non-constant --- mean curvature function. 

Technically, these Bonnet surfaces are obtained by eliminating from the Gauss-Codazzi equations the phase of the complex-valued Hopf factor $\calJ$.
After an appropriate conformal transformation (more on this below),
and up to an arbitrary choice of the origin of the local  coordinates,
one finds that the sole novel surfaces, real-valued and  analytic, are  of three different types (Appendix B of \cite{RC2017} providing a  concise but exhaustive summary). This classification essentially depends on whether a certain real parameter ---  purposefully denoted here    
$\nu^2$, as in the kernel (\ref{defKtanu}) --- 
is 
negative, positive, or zero, three    cases  referred to much later 
\cite{Cartan1942} by \'E.~Cartan as respectively A, B, and C. 

With that $\nu$ being therefore either purely imaginary or real, 
let us   define  the following  complex-valued function of $(z,\zb)$, the second equivalent formula  being mere (!) trigonometry:
 \be
 \label{defJnu} 
 \calJ_{\nu^2}(z,\zb) \coloneqq  
  \nu \coth{\lp \nu z/2 \rp} -  \nu \coth{\lp  \nu \RE{z}  \rp}\equiv  \frac{ \nu }{\sinh{\lp \nu \RE{z}  \rp}}\frac{\sinh{\lp { \nu  \zb/2}\rp}}{{\sinh{\lp {\nu  z/2}\rp}}}.
\ee 
One of the benefits of this algebraic expression is to treat all cases  on the same footing: the hyperbolic lines simply 
become ordinary trigonometric ones when goes from $\nu^2 >0$ to $\nu^2<0$, with the understanding that the 
remaining value $\nu=0$ (case C in Cartan's classification)  
at the intersection of these two 
conditions
can be  obtained by passing to the (well-defined) ``rational" limit $\nu \to 0$   in  the formula above (so that  $\calJ_{0}(z,\zb) = \frac{2}{z}-\frac{2}{z+\zb}$), and all subsequent ones.

Of course, to be able to use that function $\calJ_{\nu^2}(z,\zb)$ as the Hopf factor with a properly defined 
domain for the conformal variables, 
   the   denominators in  (\ref{defJnu})
 have  to be always nonzero. 
  One  thereby gets  two    admissible  regions, separated by and symmetric with respect to the imaginary axis $\RE{z}=0$. 
  These two \textit{mirror} regions    consist
 into a pair of infinite \textit{strips} (degenerating into the right and left half-plane 
 when $\nu^2 \to 0$), 
 which are  vertical or horizontal ones depending whether $\nu^2$ is negative or positive,   and for which the 
   modulus/the absolute value    $|\nu|=\sqrt{\mp \nu^2}$ ($\mp$ in case A/B) serves  as   a scale factor for their width or height. 
   
 As for the conformal content of (\ref{defJnu}) briefly alluded to above, since the imaginary part of the Hopf factor is that of
 an analytic function,
 it  always satisfies locally within each of those strips $4 \partial \op  \IM{\calJ_{\nu^2}}=0$,
 the bi-dimensional Laplace equation.
  This local, \textit{harmonic} characterization 
    is
   in fact a \textit{global} one, 
   if one recalls (or becomes aware of)  the classical result of \cite{Widder1961}, although obtained  in a completely different analytic context (we refer to \cite{RC2017}, p.~23,
   for another proof using PDEs methods  in the original,  geometric language).

Summarizing, either of these arguments shows  that the expression (\ref{defJnu}) for the Hopf factor  can be considered as essentially defined in a   \textit{global and unique} way (see \cite{Widder1961} for details), notably  for probabilistic  reasons dictated by  the underlying  representation through the Poisson kernel of the  non-negative solutions to the Laplace equation in a strip (or in any conformally equivalent domain).

A consequence of this harmonic behavior is that  on  either of those strips (or on both half-spaces) it holds that
\be
\label{HopfBonnet}
2 
\overline{\partial}\calJ_{\nu^2}(z,\zb)= 2 
\partial\overline{\calJ_{\nu^2}(z,\zb)}=  | \calJ_{\nu^2}(z,\zb)|^2,
\ee
where the squared modulus of the Hopf factor above  simply  evaluates to    
\be
\label{J2}
| \calJ_{\nu^2}(z,\zb)|^2 = \lp \frac{\nu}{\sinh{\lp  \nu \RE{z}  \rp}}\rp^2,
\ee
since  $\overline{\calJ_{\nu^2}(z,\zb)}=\calJ_{\nu^2}(\zb,z)$ by the
very definition (\ref{defJnu}) and the fact that  $\nu^2$ is always real. 
For reference, we also mention that the argument of the Hopf factor is given by the nice expression 
\be
\tan{\lp \frac{\arg{\calJ}_{\nu^2}}{2}\rp}=-\coth{\lp \frac{\nu \RE{z}}{2}\rp}\tan{\lp \frac{\nu \IM{z}}{2}\rp} 
\ee
already present in Bonnet's original work.

The output of all this is that  if one takes for 
the Hopf factor $\calJ$ precisely the function $\calJ_{\nu^2}$ given by (\ref{defJnu}),
or a multiple of it and/or of  its arguments, the two Codazzi equations  (\ref{defCodazzis}) automatically reduce to a single one,  and this in turn 
implies that  
both the mean curvature  function    \textit{and} 
the conformal metric factor
can be chosen to depend solely on $\RE{z}$ according to the pivotal relationship (see, e.g., equation~(21.d) in \cite{RC2017}),
\be
\label{relGHJ}
\frac{\dd \calH(\RE{z})}{\dd\! \RE{z}} = \lp \frac{\nu}{\sinh{\lp  \nu \RE{z}  \rp}}\rp^2  \calG^2(\RE{z}).
\ee
 Incidentally --- and this will be a crucial observation for the persistence problem ---  
 the right-hand-side above  is nothing but the   Willmore ``elastic energy" density  (\ref{Will1}).

The Gauss-Codazzi PDEs  (\ref{defGauss})--(\ref{defCodazzis}) for this class of \textit{isometry-preserving surfaces} --- 
realized by the
continuous deformation from  a representative   one as $\IM{z}$ is varied  ---  thereby reduce  to  a pair of coupled  nonlinear  first and second-order \textit{ordinary} differential equations (ODEs),
\ba
\label{Codax}
\calH' &=&  | \calJ_{\nu^2}|^2  \calG^2,  \\ 
\label{Gaux}
\frac{1}{2}\lp \log{\calG^2}\rp''&=& \calG^{-2} \calH^2 -  |\calJ_{\nu^2}|^2 \calG^2, 
\ea
the  primes  above standing for differentiation  with respect  to the (sole) independent variable $\RE{z}$
of this system. 
Note also that we have continued to employ  the same symbols $\calG=\calG(\RE{z})$ and $\calH=\calH(\RE{z})$  for  the 
   metric   and   mean curvature  functions of these {Bonnet surfaces}.
   The use of $\RE{z}$ (instead of $x$) for their independent variable will
   also be favored, notably because  that notation serves as a reminder of the   bidimensional nature of the conformal variables  $(z,\zb)$ in the initial Gauss-Codazzi PDEs.

Eliminating  $\calG^2$ between (\ref{Codax}) and (\ref{Gaux}), one  
immediately recovers  the third-order nonlinear ODE 
 \be
 \label{Bonnet3}
 \frac{1}{2}\lp\frac{\calH''}{ \calH'}\rp' + 
\calH' - \lp \frac{\nu}{\sinh{(\nu \RE{z})}}\rp^2 \frac{
\calH^2+ \calH'}{\calH'}=0
 \ee 
   that  Bonnet found in 1867  (\cite{Bonnet1867}, equation~(52) p.~84) for the mean curvature function of ``his" surfaces.
 It is displayed here following the standardization of \cite{CMBook2}   (equation (B.1) p.~289). The works \cite{BE1998} (equation (25) p.~55)
 and  \cite{BEBook} (equation (3.31) p.~32) have a different normalization for the  conformal coordinates and the Hopf factor in the corresponding (\ref{HopfBonnet})--(\ref{relGHJ}), implying in particular that the sign of their  $\calH'$ is always negative, thus opposite to ours, cf. (\ref{relGHJ}) here.
  
At the end of the 19th century, Hazzidakis \cite{Hazzi1897} found for   that third-order ODE
(\ref{Bonnet3}) a ``first integral",
\be
\label{Hazzi} \lp \frac{\calH''}{2 \calH'}+\nu \coth{(\nu \RE{z})} \rp^2  +  
\calH'+\lp \frac{\nu}{\sinh{(\nu \RE{z})}}\rp^2 \frac{
\calH^2}{\calH'} +  2 \nu \coth{(\nu \RE{z})}  \, 
\calH = \nu^2 \vartheta^2, \ee
 where $\vartheta^2$ is the corresponding integration constant, an always real yet possibly negative quantity, 
 i.e. $\vartheta$ can be purely imaginary,
 at par with our convention for $\nu^2$. 
Note that a simple means of checking (\ref{Hazzi}) is to multiply (\ref{Bonnet3}) by the integrating factor $\calH''/\calH'+2 \nu \coth{(\nu \RE{z})}$, or conversely to differentiate (once exhibited!) the second-order equation (\ref{Hazzi}):  (\ref{Bonnet3}) would appear, up  
 to that overall   multiplying factor.  
 
Once the (arbitrary) origin of the conformal coordinates is fixed, 
one can check using  
(\ref{relGHJ})--(\ref{Hazzi})  that the mean curvature 
  $\RE{z} \mapsto \calH(\RE{z})$ and the metric factor $\RE{z} \mapsto \calG(\RE{z})$    are     respectively odd and even.
   Overall   signs
  or scaling factors for these two real 
  functions defined on $\RR \setminus \{0\}$ are otherwise arbitrary  and can be changed at will (for generic conventions, see, e.g., \cite{RC2016}) as long the  associated joint equality (\ref{relGHJ}) is maintained. 
  One thereby essentially obtains the same  \textit{real} Bonnet surface, 
  which in general depends on six parameters \cite{Bonnet1867, Chern, RC2017}, 
  and we shall always consider representatives differing  by such normalizations as equivalent. 
  
  Nevertheless, beyond  a local description, 
   the   existence  of the global solution to the strongly nonlinear ODE (\ref{Hazzi}) and  its properties  has been
  considered  for many years by the best geometers \cite{Bonnet1867, Cartan1942, Chern} as a difficult  if not  insurmountable problem.

Eventually in 1994 \cite{BE1998},
 Bobenko and Eitner  
 succeeded in identifying   (\ref{Bonnet3})--(\ref{Hazzi})
 as 
 a certain codimension-one $\PVI$ tau-function,
 the sole real parameter $\vartheta^2$ for the latter being determined by   Hazzidakis' first integral.
 Their key observation was to  recognize in
 the Weingarten moving frame equations      an appropriately-gauged Jimbo-Miwa $2\times 2$ matrix Lax pair. This followed the pioneering work  \cite{Bobenko1994} by the first of these authors, where
 concepts and methods  from integrable systems and soliton theory were  fruitfully 
 imported into differential geometry (see also \cite{Konop1996} for related ideas).

 A noteworthy consequence of this crucial correspondence --- which took therefore more than a century to be unveiled --- is that the global existence of a Bonnet surface is now  ensured --- at least conceptually! In practice see  Appendix \ref{AppendixBOconnection} --- through Jimbo's famous solution of the tau-function $\PVI$ connection problem \cite{Jimbo1982}.

Since this remarkable discovery, the  mapping of   
 the mean curvature function of Bonnet surfaces 
 to a  particular instance (see (\ref{OJMh})) of what is called in the modern literature the 
Okamoto\-{--}\-Jimbo-Miwa sigma-form of $\PVI$
(although it first appeared as expression $t$ p.~341 in \cite{ChazyThese} 
has been detailed at length in  articles, lecture notes, or even books \cite{BE1998,BEBook,BETrieste,CMBook2}. We also mention that the extrapolation to the generic four-parameter $\PVI$ gives, due to its intrinsic geometric origin, probably the ``best" Lax pair for $\PVI$ \cite{RC2016,RC2017}.

Some properties of the Bonnet $\PVI$ mean curvature function are recalled in the (long)  Proposition \ref{PropBonnetH} below. 
This is mainly to set the stage for our Proposition \ref{PropBonnetG}, the core of this Appendix, 
which addresses  the $\PVI$ solution of the second-order ODE satisfied by the Bonnet metric function $\calG=\calG(\RE{z})$ itself
\be
\label{EG2}
\lp {\calG}'' -2 \nu^2 \calG^3 -\nu^2 \vartheta^2 \calG\rp^2+  \lp 2\nu \coth{(\nu \RE{z})}\, \calG \rp^2 \, \lp{\calG}'^2-\nu^2 \calG^4- \nu^2 \vartheta^2 \calG^2 \rp=0.
\ee

Although it is immediate to 
 obtain this equation by eliminating the mean curvature function $\calH$ (and  its derivatives $\calH', \calH''$) between  (\ref{Codax}), (\ref{Gaux}), and  (\ref{Hazzi}), to the best of our knowledge 
 (\ref{EG2}) does not seem      to
have been  written down before. 
One of the reasons might be that (\ref{EG2}) has to be identified as a particular incarnation  in this geometric context of the  rather esoteric  Chazy 
$\CVI$ equation, of which the general solution was made explicit only in 2006 \cite{CosgroveChazy2006}. 

It turns out that an equation such as (\ref{EG2}) is 
in fact \cite{CMBook2}
 ``half-way''   between the usual $\PVI$ equation and its sigma-form/tau-function, since
 it can also be viewed  as the  second-order second-degree nonlinear ODE  satisfied  by a properly shifted and rescaled \textit{impulsion} in Okamoto's 
$\PVI$ Hamiltonian formalism 
(the solution of $\CVI$ is recalled in Appendix \ref{AppendixOkamoto}).

Last but not least, there is also an additional layer of non-trivial Painlev\'e relationships holding here due to the specific monodromy exponents  of the Bonnet family. Indeed, it happens  that the solution to (\ref{EG2}) can also be   expressed algebraically  in terms of the Bonnet-$\PVI$ mean curvature function through a so-called  \textit{folding} $\PVI$ transformation, 
here a \textit{quadratic} one. 

This even more 
recondite $\PVI$
 symmetry is present only under a certain codimension-two constraint on the four monodromy exponents, namely when either two of them are zero or when there exists two couples of pairwise equal ones.
This symmetry is in fact a generalization at the level of $\PVI$ of the so-called Goursat transformation for   Gauss hypergeometric function ${}_2F_1 \lp a,b,c;x \rp$  when the  parameters and  the independent  variables are constrained according to 
\be
{}_2F_1 \lp a,b,\frac{a+b+1}{2};x \rp = 
{}_2F_1 \lp \frac{a}{2},\frac{b}{2},\frac{a+b+1}{2}; 4x(1-x) \rp.
\ee

The    folding quadratic transformation for $\PVI$ was first found  (independently)  by Kitaev and Manin in the 1990's, until it was  generalized and extended  for all Painlev\'e transcendents
  and recast using methods of algebraic geometry   by Tsuda, Okamoto, and Sakai  in \cite{TOS2005}. 
   The  most general  formulas for these quadratic and the quartic $\PVI$ transformations are recalled in Appendix \ref{AppendixOkamoto}, along with  their Hamiltonian signification as generalized time-dependent canonical transformations. 
   
Finally, recall that Bonnet surfaces are so-called Weingarten surfaces
\cite{Raffy1891} since the mean curvature and the metric are 
related by (\ref{relGHJ}), 
one has $\dd \calG \wedge \dd \calH=0$. 
One of the by-products of our findings is therefore to provide explicit    relationships showing how this is realized at the $\PVI$ level.  

\begin{Proposition}{\textbf{The mean curvature of  Bonnet surfaces as a $\PVI$ tau-function}}
(Bobenko-Eitner
\cite{BE1998, BEBook}, see also \cite{RC2017})
\label{PropBonnetH}   

 For all types of $\PVI$ Bonnet surfaces as parametrized by $\nu^2$ and $\vta^2$, the global solution  to their mean curvature function obeying (\ref{Hazzi}) can be obtained 
 through \be
\label{solBonnetH}
\calH( \RE{z}) = 2 \, \varepsilon \nu \,  h_{\vta^2}(s_\varepsilon(\nu \RE{z})), \quad s_\varepsilon(\nu \RE{z}) \coloneqq 
\frac{1- \varepsilon \coth{\lp \nu \RE{z} \rp}}{2}, \quad \varepsilon^2=1.
\ee

 With respect to either branch $\varepsilon=\pm1$ of the independent variable  
$s=s_\varepsilon(\nu \RE{z})$, the  function $s \mapsto h=h_{\vta^2}(s)$ 
obeys
 the  second-order second-degree nonlinear 
 ODE 
\be
\label{OJMh}
\lp s(s-1)\frac{\dd^2 h}{\dd s^2} \rp^2 + 4 \frac{\dd h}{\dd s} \lp s \frac{\dd h}{\dd s} -h \rp\lp(s-1)\frac{\dd h}{\dd s}-h \rp -\vartheta^2 \lp \frac{\dd h}{\dd s}\rp^2=0.
\ee
On either symmetric admissible region of the conformal coordinates, this
defines two \textit{mirror}-like Bonnet surfaces  $\calH(-\RE{z}) = -\calH(\RE{z})$ having the same
   metric 
\be
\label{solBonnetG}
\calG^2(\RE{z}) = \frac{\dd h_{\vta^2}(s)}{\dd s}\Big{|}_{s=s_\varepsilon(\nu \RE{z})}.
\ee

 This reduced mean curvature function 
coincides with
an auxiliary Okamoto  Hamiltonian evaluated
on the solutions $q=q(s)$, $p=p(s)$ of the equations of motion for a $\PVI$ Hamiltonian $\HVI(p,q,s)$ with monodromy exponents $\upabcd$:
\be
h(s)=s(s-1) \HVI(p(s),q(s),s), \quad \frac{\dd q}{\dd s}= \frac{\partial \HVI}{\partial p}, \quad \frac{\dd p}{\dd s}= -\frac{\partial \HVI}{\partial q}.
\ee
 This identification $h \equiv h_{\vta^2} \equiv h$
 takes place 
  for a co-dimension  one family of monodromy exponents  having  any signed values  
 among the
 two  sets i) or ii) parameterized by $\vartheta^2$:
\be
\label{tjsBonnet}
\lp \va^2, \vb^2, \vc^2, (\vd-1)^2 \rp 
= \begin{dcases*}
 (\vartheta^2,0,0, \vartheta^2) \quad \mathrm{set \, i)} \\   (0,\vartheta^2,\vartheta^2,0) \quad \mathrm{set \, ii)}
\end{dcases*}
\ee

The  birational equivalence existing between $(s,h,\frac{\dd h}{\dd s})$ and the $\PVI$ function $q=q(s)$  is simpler for the
set ii)
\be
\label{biraBonnetII}
\frac{h}{q-s} = -\frac{\dd h}{\dd s}, \quad h= h_{\vta^2}(s),
\ee
where that  $q=q_{\vta^2}(s)$  solves the $\PVI$ equation with coefficients that are even in $\vartheta$, as the equation for the reduced mean curvature:
\be
\label{coeffsII}
\abcd= \lb \frac{\va^2}{2}, -\frac{\vb^2}{2}, \frac{\vc^2}{2}, \frac{1-\vd^2}{2}\rb_{(\protect{\ref{tjsBonnet}}).\mathrm{ii)}}=
\lb 0,-\frac{\vartheta^2}{2},\frac{\vartheta^2}{2},0 \rb.
\ee

Conversely, given  a branch  of that  $\PVI$ function $s \mapsto q_{\vta^2}(s)$,   there exists an expression of  $h=h_{\vta^2}(s)$ in terms of the logarithmic derivative of the Chazy-Malmquist tau-function $\tau=\tau_{\vta^2}(s)$, 
 whose unique movable singularity is a simple pole of residue unity, and which is
 even in  both the first-order derivative of $q_{\vta^2}(s)$ and in  the monodromy exponent $\vartheta$:  
\be
\label{tauBonnetII}
h_{\vta^2}(s) = 
s(s-1)\frac{\dd \log{\tau_{\vta^2}(s)}}{\dd s}
=\frac{1}{4}\, \frac{s^2 (s-1)^2 }{q_{\vta^2}(q_{\vta^2}-1)(q_{\vta^2}-s)}\lp \frac{\dd q_{\vta^2}}{\dd s} \rp^2 - \frac{\vartheta^2}{4} \, \frac{q_{\vta^2}-s}{q_{\vta^2}(q_{\vta^2}-1)}.
\ee

\end{Proposition}

Since the contents of   Proposition  \ref{PropBonnetH} are well known,
we recall some elements of its proof     after asserting  Proposition  \ref{PropBonnetG}. This is mainly to pave the ground for  the demonstration of the latter,  where  the  $\PVI$ solution for the metric factor of Bonnet surfaces is exhibited.

 \begin{Proposition}{\textbf{A Chazy $\CVI$ solution for the metric factor of $\PVI$ Bonnet surfaces}}
 \label{PropBonnetG}

The solution of the ODE (\ref{EG2}) giving   the metric of Bonnet surfaces 
 can  be directly expressed 
  as
  \be
\label{defGG}
\calG^{2}(\RE{z})= -g_{\vava}^2(t(\nu \RE{z})), 
\ee
where the independent  variable $t=t(\nu \RE{z})$ for this one-parameter  Chazy $\CVI$ function $g=g_{\vta^2}(t)$ 
in rational form is also the one for a $\PVI$ function $Q(t)$ with 
the ``folded" set of 
monodromy exponents  
\be \label{thetaQ}
\upabcd = \la \frac{1-\vartheta}{2},\frac{1-\vartheta}{2}, \frac{\vartheta}{2}, \frac{\vartheta}{2}\ra, 
\ee
and a conjugate impulsion $P(t)$,  so that 
\be
\label{resgPQ}
g_{\vta^2}(t)=  \frac{2 P(t) \lp Q(t)-1 \rp \lp Q(t)-t \rp}{t-1} \equiv \frac{t}{Q(t)}\frac{\dd Q(t)}{\dd t}-\frac{1}{2}-\frac{1-\vartheta}{2(t-1)}\lp Q(t)-\frac{t}{Q(t)}\rp.
\quad 
\ee 

This $\PVI$ solution $Q(t)$  is determined by a quadratic 
transformation which intertwines it algebraically but never birationally with any of the admissible $q(s)$ (\ref{tjsBonnet}) for the mean curvature, with in particular for their 
respective independent  variables $s=s_{\varepsilon}(\nu \RE{z})$  (on either branch
 $\varepsilon^2=1$)
 and $t=t^{\pm 1}(\nu \RE{z})$ (for either sign $\pm$)
\be
s= \frac{1}{2}+\frac{1}{4}\lp \sqrt{t}+\frac{1}{\sqrt{t}}\rp, \quad \textrm{hence} \quad t(\nu \RE{z})= \coth^2{\lp \frac{\nu \RE{z}}{2} \rp} \quad \textrm{or} \quad \tanh^2{ \lp \frac{\nu \RE{z}}{2}\rp}.
\ee

The functions giving the reduced  mean curvature  and the metric factor   for the Bonnet surfaces
thus 
obey
\be
\label{g2tdhs}
\frac{\dd \hta(s)}{\dd s}\Big{|}_{s=s_\varepsilon(\nu \RE{z})} = -\gta^2(t)\Big{|}_{t=t^{\pm 1}(\nu \RE{z})}.
\ee

\end{Proposition}

\begin{proof}{(of the Proposition \ref{PropBonnetH})}.

We start by justifying the appearance of the sign $\varepsilon= \pm1$ in (\ref{solBonnetH}), 
which comes after changing $\nu \hookrightarrow \varepsilon \nu$
in  the reduced
 Gauss-Codazzi equations, or  in the resulting Bonnet-Hazzidakis ones (\ref{Codax})
 --(\ref{Hazzi}). Since all these equations  are even in $\nu$, they are of course invariant under this \textit{ad hoc} 
 operation.
 Yet, the explicit introduction of that sign is a simple  means of keeping  track of   the ex\-is\-tence of the two fundamental  regions for  $\RE{z}$ originating  from  the meromorphic ``barrier" in
 (\ref{defJnu}), and which resurfaces here    through the two branches in  the 
 mean curvature $\PVI$ independent variable $s=s_\varepsilon$. Indeed, because of   the oddness of the $\coth{(\cdot)}$ function, this underlying ``mirror'' symmetry can be realized 
 as the involution
 \be
\label{Reztos}
s_\varepsilon(\nu \RE{z}) =
\frac{1-  \coth{\lp \varepsilon \nu \RE{z} \rp}}{2} 
= \frac{1-  \varepsilon \coth{\lp  \nu \RE{z} \rp}}{2} \equiv 1-s_{\varepsilon}(-\nu \RE{z})\equiv 1-s_{-\varepsilon}(\nu \RE{z}).
\ee
 Equivalently, one could switch (or not)  
 these 
 two branches 
 changing $\nu \hookrightarrow -\nu$ and/or $\RE{z} \hookrightarrow -\RE{z}$, while keeping  the product $\nu \RE{z}$  even or odd.
Keeping that $\varepsilon$
as 
a remaining 
degree of freedom     will be very useful to achieve in  the context of the persistence problem a proper correspondence between geometric quantities and  probabilistic ones. 

Note that another way of proceeding would consist  of defining from the onset the Hopf factor (\ref{defJnu}) $\calJ_{\nu^2} \hookrightarrow \calJ_{\varepsilon,\nu^2}$ by ``pulling out"  an explicit sign $\varepsilon^2=1$, for instance according to
\be
\calJ_{\varepsilon,\nu^2}(z,\zb) =  \varepsilon
 \lp  \nu \coth{\lp \nu z/2 \rp} -  \nu \coth{\lp  \nu \RE{z}  \rp}\rp.
\ee
This would  modify the fundamental relationship  (\ref{Codax}) as $ \varepsilon \calH' =  |\calJ_{\nu^2}|^2 \calG^2$, so that an additional sign  would appear  in the Bonnet and Hazzidakis equations,  multiplying 
all the odd-terms in the mean curvature function and its derivatives. (Incidentally, Bobenko and coworkers made such a choice with $\varepsilon=-1$, so that they always have their $\calH'<0$.)
To have the simplest possible displayed equations,    we stick to our original prescription, bearing  simply in mind that whatever $\nu$, the mean curvature function $\calH(\RE{z})$ has to be odd with respect to $\RE{z}$.
 
At any rate, with the
explicit change of the independent variables (\ref{solBonnetH}), one has 
 \be
\frac{\dd s_\varepsilon(\nu \RE{z})}{\dd\! \RE{z}}= \frac{\varepsilon \nu}{2 \sinh^2{(\nu \RE{z})}} \equiv  \frac{1}{2 \varepsilon \nu} |\calJ_{\nu^2}(\RE{z})|^2,
\ee
so that   both parameters $\varepsilon$ and $\nu$, along with  the modulus of the Hopf factor 
are absorbed in the single Codazzi equation (\ref{relGHJ}) when one changes variables and functions $(\RE{z},\calH) \hookrightarrow (s,h)$. Hence one  immediately obtains (\ref{solBonnetG}), for short
$\calG^2 = h'$,  
  primes standing here and henceforth for $s$-derivatives.
 
 Plugging that   result for the metric  in  Gauss  equation (\ref{Gaux}) 
  gives
 for $h=h(s)$ the third-order ODE 
\be
\label{Eh3}
s(s-1) \lp h' \lp s(s-1)h'' \rp'-s(s-1)h''^2\rp + 2 h'\lp s(s-1)h'^2-h^2 \rp=0.
\ee
Converting  into these variables Hazzidakis' integrating factor ($\equiv -8 \nu^3 s^2(s-1)^2 h''$)   for Bonnet equation (\ref{Bonnet3}), one checks that 
\be
\text{the left-hand-side of}~(\protect{\ref{Eh3}})=
\frac{h'^3}{2 h''} \lb  \lp \frac{s(s-1) h''}{h'}\rp^2 + \frac{\lp 2h -(2 s -1)h'  \rp^2}{h'} -h'  \rb'.
\ee
 
 On the one hand, if we denote by  $\vartheta^2$   the first integral 
(possibly negative, as for $\nu^2$) 
of the   total
$s$-derivative above
on the \textit{non-singular  solutions}  (\cite{Okamoto1986}, Remark 1.1, p.~350) where that parameter-dependent function $h=h_{\vta^2}(s)$ is neither constant nor affine  (i.e. when the product $h'
h''\neq0$),  we therefore  obtain
\be
\label{Hazzih}
 \lp \frac{s(s-1) h''}{h'}\rp^2 + \frac{\lp 2h - (2 s -1)h' \rp^2}{h'} -h'= \vartheta^2.
\ee
This  expression is  
equivalent to (\ref{OJMh}) after some straightforward 
 algebraic
rearrangement. 

On the other hand, from Okamoto's  theory (see Appendix \ref{AppendixOkamoto} for a digest), one  knows that for arbitrary monodromy exponents 
the auxiliary $\PVI$ Hamiltonian evaluated on the equations of motion, namely the function 
\be
\label{redhgen}
s \mapsto h(s)=
s(s-1)\HVI(p(s),q(s),s) + \mathrm{some \, 
affine \, shift \,  in\, } s,
\ee
obeys a certain second-order second-degree nonlinear ODE (denoted $\EVI$ in \cite{Okamoto1980-II}, recalled 
as (\ref{sigmaOJMh}) in Appendix \ref{AppendixOkamoto}), that one can also  view as
a polynomial in $h'$.
Identifying  its coefficients,  one verifies that the  $\vta^2$-parameter-dependent reduced mean curvature function   $h \equiv h_{\vta^2}$  solving 
(\ref{OJMh}) does coincide with the auxiliary Hamiltonian function $h$ defined through   (\ref{redhgen})  when the monodromy exponents $\upabcd$ belong to  either set i) or ii) displayed in (\ref{tjsBonnet}) whatever the signs chosen  for the  squares occurring there, while the affine function of $s$ in (\ref{redhgen})
is  identically zero in both cases (see Proposition \ref{Propsigmaform} for details).

From now on, we shall most of the time simply write $h$ for  this $\vta^2$-parameter-dependent reduced mean curvature Bonnet function $h_{\vta^2}$, which is also equal to that auxiliary $\PVI$ Hamiltonian
$h$ with either set i) or ii) or monodromy exponents.  For both sets, note that 
their $\PVI$ coefficients satisfy $\beta+\gamma=0$:
this involution $\la s,q(s),h(s)\ra \leftrightarrow \la 1-s,1-q(1-s),-h(1-s)\ra$ is realized changing the sign   $\varepsilon$  in $s=s_\varepsilon(\nu \RE{z})$, cf. (\ref{Reztos}): it corresponds to a reversal of orientation for these two  mirror Bonnet surfaces,  leaving their metric invariant.

As for the justification of (\ref{biraBonnetII}), one knows  that for arbitrary   $\upabcd$  the  birational equivalence  $(s,h,h') \hookrightarrow q$     has a rather bulky expression (see Table R in \cite{Okamoto1986I}, repeated as (34) in \cite{RC2017}, or as (B.54) in \cite{CMBook2}).
Yet, for the second set 
(\ref{tjsBonnet}) which has   $\vd=1$, a lot of simplifications   take place, and one obtains --- independently of $\vartheta^2$, and still assuming $h'\neq 0$ --- the very simple (\ref{biraBonnetII}), where $q$ solves $\PVI$ with coefficients (\ref{coeffsII}).

Concerning the penultimate point, the tau-function formula (\ref{tauBonnetII}) for $h$  comes from the definition of the auxiliary Hamiltonian, using the expression of the  impulsion $p=p(s)$ in terms of 
the ``velocity"  $\frac{\dd q}{\dd s}$  determined by one of Hamilton's equation of motion 
 $\frac{\dd q}{\dd s}= \frac{\partial \HVI}{\partial p}$. Indeed, with the specific quadratic dependence in $p$ of the polynomial Hamiltonian (\ref{defHpqs}), the following simple affine
relationship holds,
\be
\label{dqdsvsp}
\frac{\dd q}{\dd s}= 
\frac{q(q-1)(q-s)}{s(s-1)}\lp 2 p 
-\frac{\vb}{q}- \frac{\vc}{q-1}-\frac{\vd-1}{q-s}
\rp,
\ee
where,  by the definition of the $\PVI$ coefficients
in terms of the monodromy exponents (cf. (\ref{coeffsII})),
 neither $q=q(s;\abcd)$ nor its derivative depend on the signs of the $\upabcd$, while the conjugate impulsion $p=p(s)$  given by the solution of  (\ref{dqdsvsp}) obviously depends explicitly 
 on those chosen for the last three ones.
 
 We  momentarily highlight  this dependence by writing
  $p_{\la \vb,\vc,\vd\ra}$ for the impulsion $p=p(s)$ solution of (\ref{dqdsvsp}),
 hence such that
 \be
 \label{defpbcd}
p_{\la \vb,\vc,\vd\ra} \coloneqq
\frac{1}{2}\frac{s(s-1)}{q(q-1)(q-s)}\frac{\dd q}{\dd s} + \frac{1}{2}
\lp
  \frac{\vb}{q} + \frac{\vc}{q-1} + \frac{\vd -1}{q-s} \rp.
 \ee
Looking up \cite{CMBook2}, equations~(B.42) and (B.44),   one reads  that the logarithmic derivative of the so-called there Chazy-Malmquist tau-function $\tau \equiv \tau_{\mathrm{VI, C},x}$
 can be expressed in general as
 \be
 \label{tauCxgen}
s(s-1)\frac{\dd \log{\tau}}{\dd s} = p_{\la \vb,\vc,\vd\ra} p_{\la -\vb,-\vc,2-\vd\ra} q(q-1)(q-s)  + \dots,
 \ee
 where 
 the  omitted terms $\dots$  are identically zero for the sets i) and ii). As for the correspondence of notations, our signed $\upabcd$   are (in that order) the 
 $\la \upta_\infty, \upta_0, \upta_1,\upta_x\ra$ of \cite{CMBook2},
 with the independent $\PVI$ variable denoted  there $x$ (here we have been using $s$), while the Riccati factor is $R(\upta_0,\upta_1,\upta_x) \equiv 2 
 p_{\la \vb,\vc,\vd\ra} q(q-1)(q-s)$ in terms of our impulsion (\ref{defpbcd}).

We specialize the above result for the Bonnet family, switching also
  for a while to notations  emphasizing  the  dependence on  and the parity with respect to  the monodromy exponent $\vta$. Let us therefore rename  the $\PVI$ function for set ii) as $q_{\vta^2}=q(s;\lb 0, -\vartheta^2/2, \vartheta^2/2,0 \rb)$. 
  For that  \textit{same} 
  $\PVI$ 
 function  $q_{\vta^2}$,  if   one chooses appropriate relative signs  when expressing the two monodromy exponents   $\vb,\vc$ 
  in terms of  $\upta$ (still setting  in (\ref{defpbcd}) $\vd=1$),  there exists   two conjugate impulsions $p_{\pm \upta}$, which can  
be   obtained by simply reversing the sign of $\upta$. Namely, if one defines 
\be
p_{\pm \upta} \coloneqq p_{\la \vb,\vc,\vd\ra}\Big{|}_{\vb= \mp \vartheta, \vc= \pm \vartheta, \vd=1} = 
\frac{1}{2}\frac{s(s-1)}{ q_{\vta^2}(q_{\vta^2}-1)(q_{\vta^2}-s)}\lp
\frac{\dd q_{\vta^2}}{\dd s} \pm \vartheta \frac{q_{\vta^2}-s}{s(s-1)} \rp,
\ee
 then  equation (\ref{tauCxgen}) for the Bonnet-$\PVI$ tau function $\tau \equiv \tau_{\vta^2}$ reduces to
\be
s(s-1)\frac{\dd \log{\tau_{\vta^2}}}{\dd s} = p_{+\upta} p_{-\upta}q_{\vta^2}(q_{\vta^2}-1)(q_{\vta^2}-s),
\ee
so that    (\ref{tauBonnetII})   follows immediately,  our derivation  also explaining why that formula appears under a factorized form    eventually even in both $\frac{\dd q}{\dd s}$ \textit{and} $\upta$. Note also that the case $\vta=0$ is  simply obtained   by  taking  the limit $\vta \to 0$   (all   formulas being holomorphic in $\vta$),  the two conjugate impulsions 
$p_{\pm \vta}$ thereby coalescing into a single one $p_0$.

Finally, the fact that this tau function   has only  one movable simple pole $s_i$ with residue unity, i.e. 
$\frac{\dd \log{\tau_{\vta^2}}}{\dd s} \sim \frac{1}{s-s_i}$, is  ensured by the general construction put forward by Painlev\'e  and his school (see for instance  section B.V of \cite{CMBook2} for a recap), the location $s_i$ of this  pole in the complex plane depending in general on the initial conditions. In our probabilistic context where the mean curvature is essentially the resolvent of an integral  kernel $K_{\vta,\nu}$ generating a determinantal point process, 
 it turns out that the pole location  is determined so that   the tau function $\tau_{\vta^2}$ \textit{always} givessquare rootto a well-normalized
(gap-)probability distribution function. 
This concludes the proof of this proposition.

\end{proof}

Up to notations or conventions that vary  widely 
(cf. the proof of Proposition \ref{Propsigmaform} 
for a discussion on this point and related matters in the general four-parameter case of the Okamoto--Jimbo-Miwa sigma form), 
all what precedes is already present   and  essentially well known in the classical or modern literature about the Bonnet $\PVI$ solution of the mean curvature ODE.

The most  demanding issue    is Proposition (\ref{PropBonnetG}),
which addresses  the $\PVI$ solution for the metric factor 
of the co-dimension one Bonnet $\PVI$ family.
Its existence relies on very specific properties of Painlev\'e functions 
that are little known, and even less commonly used. 
It is also a crucial step for our determination of the persistence probability distribution function, in particular to assess the genuine  \textit{transcendence} of the corresponding $\PVI$.

\begin{proof}{(of the Proposition \ref{PropBonnetG})}

Our starting point  is the equation satisfied  by the temporal derivative of the   explicitly time-dependent Hamiltonian  evaluated on the equations of motion. Following \cite{Okamoto1986I} (p.~353), the computation begins by recalling a simple but fundamental fact: using  the chain-rule, one has 
\be
\frac{\dd }{\dd s}\HVI(p(s),q(s),s)=\lb \frac{\partial }{\partial s}
\HVI(p,q,s) + \underbrace{p'  \, \frac{\partial }{\partial p}
\HVI(p,q,s)+ q' \, \frac{\partial }{\partial p}
\HVI(p,q,s)}_{0}\rb_{p=p(s),q=q(s)}
\ee
(where $p'=\frac{\dd p}{\dd s}$ and $q' =\frac{\dd q}{\dd s}$), since  the last two terms  cancel with each other  by the very definition of  
Hamilton's equations. Hence
 the overall time-dependence of the temporal derivative of the Hamiltonian  evaluated on the equations of motion can only originate  from  the built-in, explicit temporal dependence 
 $\frac{\partial H_{\mathrm{VI}}}{\partial s}$,  
which  by  Okamoto's construction is itself
 rigidly determined by demanding  its isomonodromy.

Taking also account in the definition  (\ref{redhgen}) of the auxiliary 
Hamiltonian function $h$  the pieces coming from the overall factor $s(s-1)$ 
and from the affine shift (cf. (\ref{defAshift})--(\ref{defBshift})),   one
thereby obtains for   arbitrary monodromy exponents $\upabcd$
\be
\label{hprime}
\frac{\dd h}{\dd s} = -\lp p^2 - p \lp \frac{\upb}{q} + \frac{\upc}{q-1}\rp \rp q(q-1)-\frac{(\upb+\upc)^2}{4},
\ee 
(matching  (2.2)  in \cite{Okamoto1986I}, or  (2.8) in \cite{FWNagoya2004},  since the latter authors use the  notations $b_1=\frac{\upb+\upc}{2}$, $b_2=\frac{\upb-\upc}{2}$), where $q=q(s)$ is a  $\PVI$ function with arbitrary coefficients $\abcd= \lb \frac{\va^2}{2}, -\frac{\vb^2}{2}, \frac{\vc^2}{2}, \frac{1-\vd^2}{2}\rb$, and $p=p(s)=p_{\la \vb,\vc,\vd\ra}$ is the conjugate impulsion 
determined through (\ref{dqdsvsp}).

Let us now specialize (\ref{hprime}) for  the particular case of a $\PVI$ function  and associated impulsion  with  a set of codimension-two monodromy exponents where (at least) two of them are zero. Without  loss of generality (up to a fractional linear transformation in the independent variable if necessary), we choose to parametrize them as 
\be
\label{thetaspqs}
\la 2 \va, 0, 0, 2 \vd \ra_{(p,q,s)}.
\ee
 Adding subscripts  to the notation  we have  used  so far provides   a simple  means of referring to the impulsion, the position, and the independent variable of the various  $\PVI$ that will show up. The arbitrary factors $2$  that we have also ascribed   in (\ref{thetaspqs}) for the monodromy exponents at infinity and around the independent variable will  also simplify some of the subsequent expressions.

Any $\PVI$  having  a set   of monodromy exponents  such as (\ref{thetaspqs})  where two of them are explicitly zero is amenable to  a quadratic folding transformation. Applying  the formulas recalled in Appendix \ref{AppendixOkamoto}, in particular Proposition \ref{PropTOS2} there, one finds  by direct algebraic substitution, i.e. without any differential elimination, that 
\be
\label{hpqPQ}
\text{the right-hand-side of}~(\protect{\ref{hprime}})\Big{|}_{(\protect{\ref{thetaspqs}})}= 
 -p^2 q(q-1) 
 =-\lp 2 P Q -\lp\va+\vd-\fs{2}\rp \rp^2,
\ee
with $p=p(s)$, $q=q(s)$, and where $P=P(t)$ and $Q=Q(t)$    are the respective impulsion and position   for  another
$\PVI$ with independent variable $t$ defined through
\be
\label{stot}
s=
\frac{1}{2}+\frac{1}{4}\lp \sqrt{t}+\frac{1}{\sqrt{t}}\rp,
\ee
and two couples of  pairwise equal monodromy exponents
\be
\label{thetasPQt}
\la \va, \va, \vd, \vd \ra_{(P,Q,t)},
\ee
which are therefore ``split up" and ``half-shared" compared to the set (\ref{thetaspqs}).

 The square roots in (\ref{stot}) entail that there is algebraic branching, with two possible solutions  for the independent variable $t$ of $Q(t)$, the folded $\PVI$. We conveniently denote these two solutions inverse from each other as $t^{\pm}(s)$, so that
\be
\label{tpmtos}
 t = t^{\pm 1}(s)=\lp 2 s -1 \pm 2 \sqrt{s(s-1)}\rp^2.
\ee

  Having in mind to apply for the real analytic  Bonnet surfaces this quadratic $\PVI$ folding transformation, 
  a very convenient way to take into account  
  the  two possible signs  above is  to restrict ourselves to the
 independent variable parametrization  
(\ref{Reztos}) for $s=s_\varepsilon(\nu \RE{z})$, since this amounts to write for $t=t(\nu \RE{z})$ (or equivalently its inverse)  
\be
\label{t2z}
t= t(\nu \RE{z})= \coth^2{\lp \frac{\nu \RE{z}}{2} \rp} \quad \mathrm{or} \quad 
\tanh^2{\lp \frac{\nu \RE{z}}{2} \rp}. 
\ee
Of course,  the relation (\ref{stot}) is more generally  valid for $s,t \in \CC$, since it is a simple transform of the famous Joukowski conformal mapping $z \hookrightarrow w=w(z)=\frac{1}{2}\lp z+\frac{1}{z} \rp$.

To justify (\ref{t2z}), note that whatever be the sign  chosen for  $\nu= \pm 
\sqrt{\nu^2}$ (in case B say) and the choice of region of $\RE{z}$ (determined by the Hopf factor (\ref{defJnu})), 
the two formulae (\ref{stot}) and (\ref{t2z}) are always equivalent 
because of the   coth-duplication formula (or cot if $\nu^2<0$ in case A) 
\be
 2 \coth{(\nu \RE{z})}= 
\coth{\lp \frac{\nu \RE{z}}{2}\rp} + \tanh{\lp \frac{\nu \RE{z}}{2} \rp}.
\ee
As for the  two possible choices 
for the variable $t=t^{\pm 1}(\nu \RE{z})$, they come from inverting the square root in (\ref{stot}) through (\ref{tpmtos}).
That sign is  independent from  $\varepsilon$, the one which keeps track in (\ref{solBonnetH}) of the orientation of the underlying surface, since the expressions in (\ref{t2z}) are even in $\nu$ (and  $\RE{z}$), hence unchanged if
$\nu \hookrightarrow \varepsilon \nu$ (and/or $\RE{z} \hookrightarrow -\RE{z}$).  Henceforth we shall always assume  that the square root  is determined  by
continuity from the origin on the positive real axis, demanding  that  $\sqrt{t^{-1}(\nu \RE{z})}= \tanh{(\nu \RE{z}/2)}>0$ when $\nu^2>0$ and $\nu \RE{z}>0$.

 The crucial point now 
 is to observe that some significant
 simplification   takes place  when  the co-dimension two set of monodromy exponents (\ref{thetaspqs}) and  (\ref{thetasPQt}) for  these folded $\PVI$ obey the additional condition
\be
\label{D4quad}
2 (\va +  \vd) -1=0,
\ee
and that this relation can be fulfilled  for the  reduced mean curvature 
function $h(s)=h_{\vta^2}(s)$ of the Bonnet surfaces family parametrized by  (the square of) a single monodromy exponent.

Before we embark on detailing the algebra, let us  make a more conceptual remark about the meaning of (\ref{D4quad}), which will be also decisive to prove the genuine transcendence of the $\PVI$ occurring in the persistence problem. 

What happens is that not only the relation (\ref{D4quad}) transforms (\ref{hpqPQ}) into the much simpler  $h'= -(2 P Q)^2$, but above all it is the fingerprint of  the deep algebraic structure  at the heart of the $\PVI$ 
symmetries, namely the affine 
 $D_4$ 
 Weyl root system. 
 Indeed, the constraint (\ref{D4quad}) corresponds to sit on the \textit{chamber of the wall} $\varTheta=0$ of this root system, since 
 for generic monodromy exponents
 this parameter $\varTheta$ is defined through the affine relationship
 \be
 \label{D4gen}
\va+\vb+2\varTheta+\vc+\vd=1.
 \ee 
 
We recall (see Appendix \ref{AppendixOkamoto}, or e.g. \cite{FWNagoya2004} for more details) that reflections in the above   affine \textit{theta}-parameter space
 are  realized  as birational  transformations between solutions, 
 precisely with the polynomial Hamiltonian constructed by Okamoto   
 to uncover this $D_4$ symmetry.
 In particular,
from a given ``seed" $\PVI$ solution $Q=Q(t)$ with 
monodromy exponents $\la \vartheta_j \ra_{j}$ such that $\varTheta \neq  0$,  the   elementary Okamoto transformation   generates through $Q \hookrightarrow   Q \pm  \varTheta/P$ another $\PVI$  (actually, 
 taking care of all possible signs, up to $2 \times 2^4=32$ contiguous ones),
 with shifted
  monodromy exponents   
  $ \la \vartheta_j  \pm \varTheta \ra_j$. 
 
 Conversely, the condition $\varTheta=0$  is    \textit{necessary}   --- but not  sufficient  --- to give rise
 to  the so-called \textit{classical} solutions of $\PVI$, which are either Riccati-transformed of the hypergeometric function ${}_2F_1$ (as the two-parameter family encountered in \cite{FWNagoya2004}), or algebraic ones (such as $t \mapsto \sqrt{t}$). 
 In practice, hypergeometric solutions can be obtained by demanding that in the Hamiltonian formalism their
  impulsion be $P \equiv 0$. Since the metric function for Bonnet surfaces is directly proportional to the impulsion of the associated folded $\PVI$, cf. (\ref{resGghPQ}), the condition $P=0$ has therefore to be scrutinized in detail.
  
  It  happens that for the Bonnet surfaces encountered in this work 
 and coming from a determinantal point process generated by  the probability convolution kernels (\ref{defKtanu}),   both classical, hypergeometric
   $\PVI$ solutions and transcendental ones are realized, 
 and that the Bonnet-$\PVI$  (of type B) occurring in the persistence problem  is  genuinely transcendental, with $P=P(t) \neq 0$.

For the moment, it remains to  check that one can find a set of (signed)  monodromy exponents holding for the  Bonnet-$\PVI$ mean curvature function which also satisfies the  codimension-one constraint (\ref{D4quad}). This is indeed the case  
for   the set i), which recalling      (\ref{tjsBonnet})  corresponds to
$\va = \varepsilon_\alpha \vta$,  $\vd=1 + \varepsilon_\delta \vta$ (and $\vb=\vc=0$)
if  one  chooses opposite relative signs $\varepsilon_\alpha \varepsilon_\delta=-1$ 
for the monodromy exponents at infinity and around the independent variable. 
The sole resulting sign can therefore be ascribed   through  $\vta$, and this gives rise to a pair of sets having suitable  monodromy exponents 
  for our purpose, to wit
\be
\label{thetaset1}
\lvabcd= \la \varepsilon_\vta \vta, 0, 0, 1- \varepsilon_\vta \vta \ra_{(p,q,s)}, \quad 
\varepsilon^2_\vta=1.
\ee

Indeed, such a choice entails that the 
 $s$-derivative of the 
reduced mean curvature
  $h=h_{\vta^2}(s)$ --- the latter function being  by construction  the same  either for set i) or ii) --- is \textit{also} directly related  through a quadratic transformation
 to another   
 $\PVI$  with  impulsion $P(t)$ and position   $Q(t)$ 
 according to
 \be
 \label{DhsPQt}
 \frac{\dd h_{\vta^2}(s)}{\dd s}\Big{|}_{s=s_\varepsilon(\nu \RE{z})} = -4 P^2(t)Q^2(t) \Big{|}_{t=t^{\pm 1}(\RE{z})}.
 \ee
This $\PVI$ 
has therefore   two couples of pairwise equal 
monodromy exponents  $\va=\vb$ and $\vc=\vd$ solely parametrized by $\vta$
 \be
 \label{resthetaG1}
\la \va, \vb, \vc, \vd \ra = \la \frac{\vta}{2}, \frac{\vta}{2}, \frac{1-\vta}{2}, \frac{1-\vta}{2} \ra_{(P, Q, t)},
 \ee
or, and   equivalently due to (\ref{thetaset1}),
the  set obtained  by changing  everywhere above
$\vta \hookrightarrow -\vta$.

Now we combine
our definition (\ref{defGG}) of  the reduced metric function $g=g_{\vta^2}(t)$ for a Bonnet surface, 
along with the results which on  the one hand  express $\calG^2= \calG^2(\RE{z})$ 
   differentially in terms of   the reduced mean curvature  $h=h_{\vta^2}(s)$, and on the other hand that same quantity algebraically in terms of the impulsion and the position of the folded $\PVI$
   with independent variable $t$, in order  to arrive at our fundamental relationship (\ref{g2tdhs})
 \be
 \label{resGghPQ}
 \calG^2(\RE{z}) \stackrel{(\protect{\ref{solBonnetG}})}{=} \frac{\dd h_{\vta^2}(s)}{\dd s}\Big{|}_{s=s_\varepsilon(\nu \RE{z})} \stackrel{(\protect{\ref{defGG}})}{=} - g^2_{\vta^2}(t)\Big{|}_{t=t^{\pm 1}(\nu \RE{z})}
   \stackrel{(\protect{\ref{DhsPQt}})}{=}
 -4 P^2(t)Q^2(t) \Big{|}_{t=t^{\pm 1}(\nu \RE{z})}.
 \ee
Hence  for short $g= \pm \sqrt{-h'}= \pm \sqrt{4 P^2 Q^2}$, 
so that   we obtain  (choosing say  $+$ signs everywhere when taking the square roots)
\be
\label{resgPQ1}
g_{\vta^2}(t)  =
2 P(t) Q(t) = \frac{(t-1)Q}{(Q-1)(Q-t)}
\lb \frac{t}{Q}\frac{\dd Q}{\dd t}
-\frac{1}{2} -\frac{\vta}{2(t-1)} 
\lp Q - \frac{t}{Q} \rp \rb,
\ee
where  the last equality comes from (\ref{defpbcd}), after  expressing for the set of monodromy exponents (\ref{resthetaG1}) the corresponding
$P=P(t)$ 
in terms of  $Q=Q(t)$ and of its derivative $\dd Q/\dd t$, along with    some algebraic rearrangement.

A somewhat more compact  formula  ---  the  one given as (\ref{resgPQ}) in the lemma  --- can be arrived at  if one uses the permutation symmetry (denoted $\pi_2(x)$ in Table~(3.3) p.~723 of \cite{TOS2005}, and $r_1$ in Table~1 p.~43 of \cite{FWNagoya2004})  which  exchanges $\va \leftrightarrow \vd$ and $\vb \leftrightarrow \vc$ while preserving the independent variable in Okamoto's Hamiltonian 
according to the birational transformation 
\be
\label{homoa2d}
(P,Q,t)_{ \vabcd} \hookrightarrow \lp \widetilde{P} = -\frac{P(Q-t)^2-\varTheta(Q-t)}{t(t-1)}, \widetilde{Q}=\frac{(Q-1)t}{Q-t}, \widetilde{t}=t \rp_{\la \vd, \vc,\vb,\va \ra}.
\ee
We apply this involution for the particular set (\ref{resthetaG1}) sitting on the chamber $\varTheta=0$ of the $D_4$ root system. The condition (\ref{D4quad}) is (of course) preserved under the permutation $\va \leftrightarrow \vd$, and one obtains
\be
\label{resgPQ2}
g_{\vta^2}(t)  = - \frac{2 \widetilde{P} \lp \widetilde{Q}-1\rp\lp \widetilde{Q}-1\rp}{t-1}= - \lb \frac{t}{\widetilde{Q}}\frac{\dd \widetilde{Q}}{\dd t}
-\frac{1}{2} -\frac{1-\vta}{2(t-1)} 
\lp \widetilde{Q} - \frac{t}{\widetilde{Q}} \rp \rb,
\ee
which is (\ref{resgPQ}) as claimed (up to the notational tildes and the irrelevant global sign), 
the final expression between brackets being also the same as  (\ref{resgPQ1}) after exchanging $\vta \leftrightarrow 1-\vta$, since for this permutation $2\va \leftrightarrow 2\vd$.

Let us also remark  \textit{en passant}, and as already  hinted at after (\ref{D4gen}),  the existence of the elementary (zero-parameter) $\PVI$ algebraic solution $Q(t)= \pm \sqrt{t}$, which  \textit{a priori} is always allowed 
here  with a couple of pairwise equal monodromy exponents such as (\ref{thetasPQt}), whatever their values  (cf. \cite{CMBook2},  equation~(B.193), p. 344). Yet,  such  an  admissible solution of the $\PVI$ equation would  automatically make  the reduced metric function identically zero,  $g_{\vta^2}(t) \equiv 0$, by
 separately ``annihilating"   in a differential or algebraic fashion (and whatever $\vta$) the first two ($=\frac{t}{Q}\frac{\dd Q}{\dd t}
-\frac{1}{2}$) and  the last two terms ($=Q -\frac{t}{Q}$)  appearing within the  brackets of (\ref{resgPQ1}) (and similarly for (\ref{resgPQ2})). Furthermore, this algebraic solution, as all the others of this nature, has zero-parameter left to accommodate for the initial conditions, and it has therefore to be discarded here, since it cannot take into account the thinning parameter occurring in our kernels.

Finally, under the change of the independent and dependent variables $(\RE{z},\calG^2) \hookrightarrow (t, -g^2)$ (for either choice of  $t=t^{\pm 1}(\nu \RE{z})$), one can  check that  the ODE (\ref{EG2}) for the Bonnet metric  is transformed into  the 
following one-parameter
 instance  of Chazy $\CVI$ equation in rational coordinates, viz.
\be
\label{CVgtBonnet}
\frac{1}{2}\lb \frac{\dd^2 g}{\dd t^2} +\lp \frac{1}{2 t}+ \frac{1}{t-1}\rp 
\frac{\dd g}{\dd t}+\frac{2 g^3 -\vartheta^2 g}{t(t-1)^2} \rb^2 
=   2\lp \frac{1}{2 t}-\frac{1}{t-1} \rp^2 g^2\lb \lp\frac{\dd g}{\dd t}\rp^2+ \frac{g^4 -\vta^2 g^2}{t(t-1)^2}\rb.
\ee

In this one-parameter situation, one is thereby led to 
the equation (\ref{CVIgenrat}).
Namely, plugging in (\ref{CVgtBonnet}) either expression (\ref{resgPQ1}) or (\ref{resgPQ2}) for $g= \pm g_{\vta^2}(t)$ (the sign being irrelevant, since (\ref{CVgtBonnet}) is even in $g$), and  using the 
  $\PVI$ equation for $Q(t)$ (and correspondingly for $\widetilde{Q}(t)$) to express the  second and third-order derivatives $\dd^2 Q/\dd t^2$, $\dd^3 Q/\dd t^3$  in terms of the powers of $\dd Q/\dd t$,  one obtains
  a polynomial   in $\dd Q/\dd t$ (or $\dd \widetilde{Q}/\dd t$),   of degree six (and without constant term), which has to be identically zero.  By identification of  its resulting  (huge!) coefficients, the 
  scrupulous 
  reader can verify that
  (\ref{CVgtBonnet}) is  indeed identically satisfied when the monodromy exponents of the corresponding $\PVI$'s  have  respective  values (\ref{resthetaG1}) or (\ref{thetaQ}). 
  Note finally that (\ref{CVgtBonnet}) is even in $\vta$, even though the formulas (\ref{resgPQ1}) or (\ref{resgPQ2}) break 
  explicitly the parity in the monodromy exponent.

  This terminates the long proof of this  Proposition \ref{PropBonnetG} and concludes  this Appendix.
\end{proof}

\section{Three remarkable second order nonlinear ODEs: $\PVI, \HVI, \CVI$}
\label{AppendixOkamoto}
\label{Appendix-PVI-HVI-CVI}
\setcounter{equation}{0}\renewcommand\theequation{C\arabic{equation}}

This appendix summarizes the links between three remarkable second order nonlinear ODES,
and recalls the folding transformation of $\PVI$.

The first of these ODEs, the sixth Painlev\'e equation $\PVI$ 
\ba
w'' &=& \frac{1}{2}\lp \frac{1}{w}+\frac{1}{w-1}+\frac{1}{w-z}\rp w'^2
-\lp \frac{1}{z}+\frac{1}{z-1}+\frac{1}{w-z}\rp w'\\
\label{defP6wz}
\nonumber
&+&\frac{w(w-1)(w-z)}{z^2(z-1)^2}\lp \alpha + \beta \frac{z}{w^2}+\gamma \frac{z-1}{(w-1)^2}+ \delta
\frac{z(z-1)}{(w-z)^2}\rp,
\ea
has for general solution a function, by construction also called $\PVI$,
which is transcendental for generic values of its four ``monodromy exponents'' $\vartheta_j$,
\be
\label{t2toa}
\lp \va^2, \vb^2, \vc^2, \vd^2\rp= \lp 2 \alpha,-2 \beta,2 \gamma,1-2 \delta\rp,
\ee
i.e.~not expressible in terms of solutions of ODEs of order one or linearizable.

\label{Propsigmaform}
The second ODE is obeyed by a Hamiltonian of $\PVI$ \cite{MalmquistP6},
whose position variable $q$ is equal to the $w$ of $\PVI$,
suitably normalized so as to obey a very symmetric second degree ODE \cite{ChazyThese},
the so-called Okamoto--Jimbo-Miwa sigma-form of $\PVI$ \cite{Okamoto1980-II} \cite{JM1981},
\be
\label{sigmaOJMh}
h'\lp s(s-1) h'' \rp^2 + \lp h'^2 -2 h'(s h'-h) + b_1 b_2 b_3 b_4 \rp^2 - \prod_{j=1}^4 (h'+b^2_j) =0, 
\ee
\be
\label{t2b}
2 b_1 = {\vb + \vc}, \quad 2 b_2 =  {\vb - \vc}, \quad 2 b_3 =  {\vd-1 +\va}, \quad 2 b_4 = {\vd-1 -\va}.
\ee
(We hope that no confusion arises with the position $q$ in the Hamiltonian formalism,
and the notation used in the Introduction for the number of states of the Potts model.)

This Hamiltonian $\HVI(q,p,s)$ is an affine transform of $h$ defined by
\ba
\label{defhVIgen}
h(s) &\coloneqq&  s(s-1)\HVI(q(s),p(s),s) + A s + B,
\\ \label{defAshift}
4 A                &\coloneqq& -\va^2+(\vd -1)^2+2 (\vb+\vc)(\vd-1) \equiv 4(b_1 b_3 + b_1 b_4 + b_3 b_4),
\\ \label{defBshift}
8 B                &\coloneqq& \va^2-\vb^2+\vc^2-(\vd -1)^2 - 4 \vb(\vd-1) \equiv -4 \prod_{1 \le j<k \le 4}b_j b_k,
\ea 
whose explicit expression is 
\be
\label{defHpqs}
 H_{\mathrm{VI}}(p,q,s) \coloneqq \frac{q(q-1)(q-s)}{s(s-1)}\lp p^2 -p\lp \frac{\vb}{q} + \frac{\vc}{q-1}+\frac{\vd-1}{q-s}
\rp
+\frac{(\va+\varTheta)\varTheta}{q(q-1)}\rp,
\ee
with the affine constraint
\be
\va+\vb+2\varTheta+\vc+\vd=1,
\ee
while the impulsion is
\be
 \label{ps2qs}
p=
\frac{1}{2}\frac{s(s-1)}{q(q-1)(q-s)} \frac{\dd q}{\dd s} 
+ \frac{1}{2} \lp \frac{\vb}{q} + \frac{\vc}{q-1} + \frac{\vd -1}{q-s} \rp,
\ee
with $q=w$, $s=z$ on the solutions of the equations of motion. 
The normalization in (\ref{sigmaOJMh}) is such that the unique movable simple pole of $\HVI$
has residue unity. 

The two ODEs for $\PVI$ and for $\HVI$ are birationally equivalent\cite[Table R]{Okamoto1980-II}.

It is useful to remark that the  Lagrangian, classically defined as
\be
\label{Lagrangiangenp}
L_\mathrm{VI}(s)  \coloneqq p \frac{\dd q}{\dd s} - H_\mathrm{VI}(p,q,s) = \frac{q(q-1)(q-s)}{s(s-1)} p^2 + \frac{\va^2 -(\vb+\vc+\vd -1)^2}{4 s(s-1)}\Big{|}_{p=p(s),q=q(s)},
\ee
possesses no term linear in the impulsion $p$ if one expresses back $\dd q/\dd s$ in terms of the impulsion $p=p(s)$ with the help of the equation of motion (\ref{ps2qs}) for the latter. In particular, this ``constant" term (viewing the Lagrangian as quadratic polynomial in $p=p(s)$) is zero on the wall of the Weyl chamber, a condition which can always be fulfilled for all the Bonnet-$\PVI$ which appear in this work.

\label{PropChazyCVI}

Finally, the third remarkable ODE, denoted $\CVI$, is the one obeyed by the impulsion $p$ of $\HVI$,
Eq.~(\ref{ps2qs}).
Rather than the normalization chosen by Chazy \cite[p.~342]{ChazyThese},
 we prefer the following one,
because it gives directly the metric function of Bonnet surfaces,
\be
\label{CVIgentrighyp}
\lp G'' + 2 G^3 + A_2 G +A_3\rp^2=  \lp 2 \coth{Z}\rp^2\lp G-\frac{A_1}{\cosh{Z}}\rp^2\lp G'^2 +G^4 +A_2 G^2+2 A_3 G+ A_4\rp.
\ee
Indeed, the explicit link with $\PVI(Q,t,\abcd)$ is 
\cite{CosgroveChazy2006,CMBook2}
\be
\label{GZ2gt}
t= \tanh^2{(Z/2)} \quad \mathrm{or} \quad \coth^2{(Z/2)},
\ee
\be
\label{solgtgen}
G =   \frac{t}{Q}\frac{\dd Q}{\dd t} 
+\frac{(Q-1)(Q-t)}{t-1}
\lp \frac{\vb}{Q}-\frac{\va+\vb-1}{2(Q-1)}- \frac{\va+\vb+1}{2(Q-t)} \rp
\ee
\ba
\label{A1t}
2 A_1 &=& -(\va -\vb), \\
\label{A2t}
2 A_2 &=& -(\va+\vb-1)^2 - 2\vc^2 -2\vd^2, \\
\label{A3t}
2 A_3 &=& -(\va+\vb-1)(\vc^2 - \vd^2), \\
\label{A4t}
16 A_4 &=&  \lb (\va+\vb-1)^2-4\vc^2 \rb \lb (\va+\vb-1)^2 - 4 \vd^2 \rb.
\ea
Note that the change of variables defined by (\ref{GZ2gt}) and 
\be
g(t) = G(Z),   
\ee
maps the nonalgebraic ODE (\ref{CVIgentrighyp}) to the algebraic one
\ba
\label{CVIgenrat}
\lb \frac{\dd^2 g}{\dd t^2}
+\lp \frac{1}{2 t}+ \frac{1}{t-1}\rp 
\frac{\dd g}{\dd t}
+\frac{2 g^3 +A_2 g +A_3}{t(t-1)^2} \rb^2
-  4\lp \frac{1}{2 t}-\frac{1}{t-1} \rp^2 &\times& \\
\nonumber
\times \lp g + A_1 \frac{t-1}{t+1}\rp^2
\lb
\lp \frac{\dd g}{\dd t}\rp^2 + \frac{g^4 + A_2 g^2 + 2 A_3 g + A_4}{t(t-1)^2}
\rb^2 &=&0,
\ea

Under two constraints among the four monodromy exponents,
there exists one algebraic transformation leaving $\PVI$ form-invariant,
first found by Kitaev \cite{Kitaev1991}.
Suitable expressions for our purpose are taken from Ref.~\cite{TOS2005}.

\begin{Proposition}{\textbf{The quadratic $\PVI$ transformation}}
\label{PropTOS2}

Let $q=q(s)$ and $Q=Q(t)$ be two $\PVI$ functions with  independent variables $s,t$,  and $p=p(s)$ and $P=P(t)$ the  associated impulsions defined  from the 
Okamoto Hamiltonian, both having   
 co-dimension two  signed monodromy exponents, 
 respectively  pairwise-constrained according to 
\be
\la 2 \vartheta_{\alpha}, 0, 0 , 2 \vartheta_{\delta} \ra \quad \text{for} \,\, p,q, \quad \text{and} \,\, 
\la
\vartheta_{\alpha},\vartheta_{\alpha},\vartheta_{\delta},\vartheta_{\delta}
\ra \quad \text{for} \,\, P,Q.
\ee
A quadratic folding transformation $(p,q,s) \hookrightarrow (P,Q,t)$ 
relates algebraically 
   the   independent variables, the positions and  the impulsions evaluated on the equations of motion according to 
\be
\label{TOSst}
s= \frac{1}{2}+\frac{1}{4}\lp \sqrt{t}+\frac{1}{\sqrt{t}}\rp,
\ee
with 
\be
\label{TOSqQ}
q(s) = \frac{1}{2}+\frac{1}{4}\lp \frac{\sqrt{t}}{Q(t)}+\frac{Q(t)}{\sqrt{t}}\rp,
\ee
and
\be
\label{TOSpP}
  p(s) = 4\sqrt{t}  \, \frac{Q(t)}{Q^2(t) -t}\lp 2 P(t) Q(t) -\lp \va+\vd-\fs{2} \rp\rp.
\ee

This algebraic transformation is also a time-dependent canonical one, 
which preserves the  respective Hamiltonian symplectic forms up to an overall factor two:
 \be
\label{symplecticII}
\dd p \wedge \dd q - \dd H_{\la 2 \va, 0,0, 2 \vd \ra} 
\wedge \dd s
=
2 \lp \dd P \wedge \dd Q - \dd H_{\la \va, \va, \vd, \vd \ra} 
\wedge \dd t
\rp,
\ee 
so that  for the  corresponding Hamiltonians   evaluated on the equations of motion 
one has 
\ba
\label{HamiltonianII}
  s(s-1)H_{\la 2 \va, 0, 0, 2 \vd \ra}(p(s),q(s),s)  = \frac{1}{\sqrt{t}} 
  H_{\la \va, \va, \vd, \vd \ra}(P(t),Q(t),t) \\ 
  \nonumber
  -
  \frac{t-1}{2 \sqrt{t}}\lb P(t)Q(t) -  \lp \va+\vd-\fs{2}\rp \lp \vd-\fs{2}\rp \rb.
\ea
Moreover, when specialized 
on the wall of the Weyl chamber,   the associated Lagrangian densities obey (denoting $q'=\dd q/\dd s$ and $\dot{Q}=\dd Q/\dd t$)
\be
\label{resLagrangiansBonnet}
L_{\la 2 \va, 0, 0, 2 \vd \ra} \lp q',  q,s\rp \dd s  
= 
2 \, L_{\la \va, \va, \vd, \vd \ra} \lp \dot{Q},Q,t \rp \dd t \quad \mathrm{if} \quad 2 (\va+\vd)=1.
\ee

\end{Proposition}


\section{The Borodin-Okounkov formula and  the solution of the 
Bonnet $\PVI$ connection problem}
\label{AppendixBOconnection}
\setcounter{equation}{0}\renewcommand\theequation{D\arabic{equation}}

 Around 2000, motivated 
  by combinatorial and  representation-theoretic issues, Borodin and Okounkov (BO)
  exhibited 
  \cite{BO2000}
     a remarkable identity, 
     allowing 
     to  express
   a generic
   Toeplitz determinant as a (discrete) Fredholm determinant  for the product of two Hankel operators.

     With hindsight, it turns out that the very same formula had also been  obtained by  Geronimo and Case as early as 1979 (see equation (VII.28) in the Appendix of \cite{GC1979}), 
     in the completely different context of inverse scattering problems  \`a la Gel'fand-Levitan-Marchenko on the unit circle, although that work went largely unnoticed
 even by experts. We refer to 
 \cite{SimonBook} (p.~343),
 \cite{WidomBook} (p.~93) and \cite{DIK2013} for  background 
 stories or  recent developments.

     The gist of  this Appendix is
  to observe that an application of  the BO formula 
  automatically grants us,
moreover in an explicit way (see the rather simple (\ref{resHtaphilarge})), with a 
      solution of the \textit{connection problem} \cite{Jimbo1982}  for  any  Bonnet-$\PVI$ tau function given 
   by the  Fredholm determinant 
   for the thinned $\Kta$ kernel, and thus in particular for the sech kernel ($\ta= \pm 1/2$). 
   
Indeed, viewing the Fredholm determinant $D_\ta([-T,T];\xi)$ for the $\Kta$ kernels (\ref{defKta})  defined in Section \ref{sectionDPP}  as 
a function of the interval length $\ell =2 T$ where this translation-invariance kernel acts on the real line (setting here and henceforth the scaling factor to
$\nu=1$), this determinant has the representation
   $$
   \calD_\theta(\ell=2 T;\xi) \ce D_\ta([-T,T];\xi) = \exp{\lp \int_0^\ell \! \dd x \, 
   \calH_\ta(x;\xi)\rp },
   $$
   where $x \mapsto \calH_\theta(x;\xi)$,
     the mean curvature function  
    for this $(\ta,\xi)$ two-parameter family of Bonnet surfaces, 
   satisfies --- thanks to the working of Sections \ref{section-closed-differential-system} and \ref{section-integration} --- the immediate generalization  with  $1/4 \hookrightarrow \ta^2$ of the Cauchy problem stated in Theorem~\ref{thm:A} 
   (and with the same $\xi$-dependent initial conditions) 
\be
\label{ODEHta}
   \displaystyle{\lp \frac{\calH''}{2 \calH'} +
\coth{x} \rp^2 + \frac{1}{(\sinh{x})^2} \frac{\calH^2}{\calH'} +2 (\coth{x}) \calH+ \calH' = \ta^2, \quad x >0.}
\ee

  What the Borodin-Okounkov formula achieves is  to  provide us with a  global expression valid for all $\ell >0$ of the Fredholm determinant $\calD_\ta(\ell;\xi)$, hence of its 
 $\ell$-logarithmic derivative
  $$
  \frac{\dd}{\dd \ell}\log{\calD_\ta(\ell;\xi)} \equiv
  \calH_\ta(\ell;\xi), \quad \ell>0.$$
  This proves in particular   the existence of
and the (finite and negative) value for the  asymptotic mean curvature of this family of Bonnet-$\PVI$ surfaces
\be
\label{Htaxilim}
\lim_{\ell \to + \infty} \calH_\ta(\ell;\xi) = \frac{1}{2}\lp
\ta^2 - \lp \frac{2}{\pi} \arccos{\lb \sqrt{1-\xi} \ \cos{\lp \pi \ta/2\rp}  \rb}  \rp^2 \rp, \quad 0\le \theta^2 <1, \quad 0 < \xi < 1.
\ee
Notice that the specialization of (\ref{Htaxilim}) to $\ta=1/2$ and  $\xi=\xi(m)=1-m^2$ gives back the expression (\ref{reskappam}) for the persistence exponent.

The unique negative solution to (\ref{ODEHta}) $ \calH_\theta(\cdot;\xi)$ decreasing on $\RR_+$ with a finite limit at infinity is therefore globally defined, and each member of this  two-parameter family may   be viewed as the analog of the Hastings and McLeod  solution of
 $\PII$ \cite{Hastings1980} appearing in all the Tracy-Widom distributions.

    Yet, a rather subtle but crucial point  that we shall elaborate on at length below
    is that   for all values of $\ta^2<1$ the equation (\ref{Htaxilim}) \textit{also} remains valid in the limiting case $\xi \to 1^-$ (and thus in particular in the $m=0$ symmetric persistence Ising case),  this essentially for probabilistic reasons --- in short, Soshnikov's theorem \cite{Sosh2000}. This is even though the Wiener-Hopf factorization for the kernels $\Kta$ formally fails there,  the corresponding symbol having a so-called \textit{Fisher-Hartwig singularity}.

The setup required to implement the BO result will be achieved in several steps.
\begin{itemize}

  \item{}  
 Switch from the thinning parameter $\xi$ to a uniformizing parameter
   $\varphi$ defined through
\be
\label{defxi}
\xi =
\frac{\cos{(\pi \theta)}-\cos{(\pi \varphi)}}{ \cos{(\pi \theta)}+1}, 
\quad \text{and  such that} \quad   0< \theta^2 < \varphi^2 < 1 \quad 
\text{for} \quad 0<\xi<1.
\ee
Note that conversely we have
\be
\label{defvarphi}
\varphi \coloneqq \frac{2}{\pi} \arccos{\lb \sqrt{1-\xi} \,.\, \cos{\lp \pi \ta/2\rp}  \rb} \equiv 
 \frac{2}{\pi} \arcsin{\lb \sqrt{\xi+ (1-\xi) \cdot \sin^2{(\pi \ta/2)}}\rb},
\ee
 these two expressions  being equivalent on the common  branch $[0,\frac{\pi}{2}]$ of the  two reciprocal trigonometric functions if $\varphi$ is chosen positive (and $<1$) in (\ref{defxi}). 
 For any $0 < \xi < 1$, we shall also assume $0< \ta$, 
 so that we always have $0<\ta <  \varphi<1$, 
 without loss of generality since the kernel $\Kta$ is even in $\ta$. 
 Most of our final expressions turn out to be enen functions of both 
 $\ta$ and $\varphi$, or this will be easy to restore. 
 As usual, the case $\ta=0$ is obtained without any problem by passing to the limit, this temporary restriction making easier some intermediate contour integral manipulations.
 
 One of the advantages of the bijective parametrization (\ref{defvarphi}) is  that the 
 \textit{Wiener-Hopf symbol} $\widehat{W}_{\taphi}$ for $\xi \Kta$, which is
conventionally defined as 
 the Fourier transform of $\Id-\xi \Kta$, has a very symmetric form in $\ta,\varphi$ \textit{and} the Fourier 
  independent variable $u$:
\be 
\label{symbolWtamu}
\widehat{W}_{\taphi}(u)\coloneqq \calF[\Id -\xi K_\ta](u)=
\frac{\cosh{(\pi u)}+ \cos{(\pi \varphi)}}{\cosh{(\pi u)}+ \cos{(\pi \ta)}}.
\ee
The normalization for our direct and inverse Fourier transforms is
\be f(x)=\int_{\RR}\! \frac{\dd u}{2 \pi} \ 
e^{-\ii x u} \widehat{f}(u),\quad
\widehat{f}(u)=\int_{\RR}\! \dd  x \  e^{\ii u x} f(x),
\ee
with the respective abbreviations 
$f=\calF^{-1}[\widehat{f}], \widehat{f}=\calF[f]$. As a means of understanding the origin of  
(\ref{defxi}), we recall in particular  (and this is a tabulated cosine Fourier transform) that for the $\Kta$ kernel
\be
\label{resFTKta}
\widehat{K}_\ta(u) = \frac{1 + \cos{(\pi \ta)}}{\cosh{(\pi u)} + \cos{(\pi \ta)}},
\ee
the value at the origin in Fourier space $\widehat{K}_\ta(0)=\int_\RR \Kta=1$ 
ensuring  the probability
normalization on the full real line of the point process generated  by this kernel.

From now on, when  $\xi<1$, thus also $\varphi <1$, we index with the subscripts $\ta,\varphi$ the quantities
of interest, so that  the mean curvature function of this   two-parameter family of Bonnet surfaces is given for all $\ell >0$ by the logarithmic derivative 
\be
\label{defHtaphiDlogDet}
\calH_\taphi(\ell) =  
\frac{\dd}{\dd \ell} 
\log \calD_\taphi(\ell), \quad 0 < \ta < \varphi <1,
\ee
identifying also $\calD_\taphi(\ell) \equiv \calD_\ta(\ell;\xi)$ thanks the bijective  parametrization $\xi \leftrightarrow \varphi$. 
 Of course,
by Fourier conjugation, the Fredholm determinant of  
the ``truncated" Wiener-Hopf operator (as it customarily referred to in the corresponding literature)   with  symbol
$\widehat{W}_\taphi$, and that for  the thinned kernel 
$\xi \Kta$ (restricted to $[0,\ell]$), are the same.

At this stage, the proof must be split into two cases, $\xi\not=1$ and $\xi=1$.
Indeed, it is only for  $\varphi<1$ that
the function $u \mapsto \widehat{W}_\taphi(u)$ given by (\ref{symbolWtamu})
 never vanishes on the real axis and is bounded (strictly) by one for all finite $u$,  therefore (\ref{resFTKta}) can be extended as a complex  analytic function 
 in the strip  $-(1-\ta) < \IM{u} < 1-\ta$. 

If and only if $\varphi < 1$ (i.e. $\xi<1$), 
the symbol $\widehat{W}_{\taphi}$
is therefore factorizable in the sense of Wiener and Hopf, and the factorization valid in the above strip is  just obtained  by inspection  using twice the 
elementary cos(h)-addition relationship, then  (four times!) the  Gamma function complement formula $\Gamma(z)\Gamma(1-z)=\pi/{\sin{(\pi z)}}$,
\be
{\widehat{W}}_{\taphi}(u)=
F_{\taphi}(\ii u/2) F_{\taphi}(-\ii u/2)
\ee
where
\be 
\label{WHfactors}
F_{\taphi}(z) \coloneqq \frac{\Gamma(\frac{1+\ta}{2}-z)\Gamma(\frac{1-\ta}{2}-z)}{\Gamma(\frac{1+\varphi}{2}-z)\Gamma(\frac{1-\varphi}{2}-z)} 
\quad 0 < \ta < \varphi <1.
\ee
The two Wiener-Hopf factors, traditionally denoted as $F^\pm$,
and which share the above-mentioned strip as a  common region of analyticity, are  thus $F^\pm(u)= F_\taphi(\pm \ii u/2)$, where 
$z \mapsto F_\taphi(z)$ is  meromorphic, with two sets of simple poles and simple zeros on the positive real axis in the $z$ variable,
which interlace but never overlap. 
Note also the nice symmetry for their inverses in the strip: 
$1/F_\taphi(\ii u/2)= F_{\varphi,\ta}(-\ii u/2)$.

  \item Iff $\xi \not=1$, application to the then factorizable Wiener-Hopf symbol
    of the explicit formula of Borodin and Okounkov \cite{BO2000},
\be
\label{BOtaphigen} 
\forall  \ell >0:\
\log{\calD_\taphi(\ell)} = \calA_\taphi \cdot \ell + \calB_\taphi + 
\log{\Det{\lb \Id - \calL_{\taphi}\rb}}\!\upharpoonright_{[\ell,+\infty)}
\ee
the three terms involving Fourier integrals as now detailed.

The first two terms of (\ref{BOtaphigen}) 
are known explicitly in terms of the inverse Fourier transform of
the logarithm of the Wiener-Hopf symbol,
\be
\label{defetaphi}
e_\taphi(x) = \calF^{-1}\lb \log{{\widehat{W}}_\taphi}\rb(x),
\ee
the corresponding Kac-Akhiezer  formulas \cite{Kac1954} being given by
\be
\label{defABtaphi}
\calA_\taphi = \lim_{x \to 0}e_\taphi(x), \quad 
\calB_\taphi= \int_0^{\infty}\! \dd x \ x \cdot  e_\taphi(x) \cdot e_\taphi(-x).
\ee
In our continuous convolution operator setting,
these two terms are  respectively the analogs of the weak and strong forms of Szeg\"o theorem for Toeplitz (discrete) determinants.

The function $e_\taphi$ (which here is even, like the kernel $\Kta$) can be computed using the residue theorem
and a suitable contour, 
which eventually amounts to summing a geometric series. 
This is the method followed in the proof of Ref.~\cite[Corollary 6]{FTZ2022}
for the persistence sech kernel.
A  shortcut exists for the more general $\Kta$ kernel, which consists in 
using the analyticity of the never vanishing symbol $\log{\widehat{W}_\taphi}$
in the parameter $\varphi$:  
if one differentiates $e_\taphi(x)$  with respect to $\varphi$, 
one recognizes in the resulting
  integral merely 
  the Fourier transform of the $K_\varphi$ kernel (up to some prefactor)
 \be
  \label{derivmue}
  \frac{\partial}{\partial \varphi} e_{\ta,\varphi}(x)= -\pi \sin{(\pi \varphi)}
  \int_\RR\! \frac{\dd w}{2 \pi}\frac{e^{-\ii x w}}{\cosh{(\pi w)} + \cos{(\pi \varphi)}}= -\frac{\sinh{(\varphi x)}}{\sinh{x}}.
  \ee
  Integrating back (\ref{derivmue}) from the obvious limiting value $\lim_{\varphi \to \ta^+} e_{\ta,\varphi}(x)=e_{\ta,\ta}(x)=0$ (corresponding to $\xi \to 0^+$) yields
  \be
  \label{resetaphix}
  e_{\ta,\varphi}(x) =  \frac{\cosh{(\ta x)}-\cosh{(\varphi x)}}{x \sinh{x}}.
  \ee
Consequently, this determines  (by continuity of the Fourier transform)  the first term of the r.h.s. of (\ref{BOtaphigen}), the so-called \textit{geometric means} for our symbol 
\be
\calA_{\taphi} \coloneqq \lim_{\ell \to + \infty} 
\frac{\log{\calD_\taphi(\ell)}}{\ell}  
=\lim_{x \to 0}e_{\taphi}(x)
= \frac{\theta^2 -\varphi^2}{2}.
\label{resAtaphi}
\ee

To obtain  the second term $\calB_\taphi$,
we use  its equivalent expression as an integral on the imaginary axis   involving
the logarithm of the  two Wiener-Hopf factors  
\be
\calB_\taphi = \int_{-\ii \infty}^{+ \ii \infty}
\! \frac{\dd z}{2 \ii \pi} \log{F_\taphi(-z)} \,  \frac{\dd}{\dd z}\log{F_\taphi(z)}.
\ee
(This is a known general formula that can be obtained 
either by a standard Fourier computation, see e.g. \cite{BE2005}, 
or as a certain Sobolev space isometry \cite{Bufetov2024}.) Because of the meromorphy of $F_\taphi$, the above integral only involves simple poles.
Following  exactly the  method of proof for  Ref.~\cite[Lemma 3.26]{BE2005},
a result for the sech kernel (and with their $\beta=\varphi-1/2$ in our notations) which, as  intriguingly remarked by Basor and Ehrhardt \cite{BE2005},
seemed  to be already in 2005 ``of interest in its own",  
we eventually obtain for the limiting amplitude $\calB_\taphi$ of the determinant the very symmetric expression 
\be
\exp{\calB_{\taphi}} =\label{resBtaphi}
\frac{  G^2\lp 1+\frac{\ta+\varphi}{2}\rp G^2\lp 1-\frac{\ta+\varphi}{2}\rp G^2\lp 1+\frac{\ta-\varphi}{2}\rp G^2\lp 1-\frac{\ta-\varphi}{2}\rp}{G(1+\ta)G(1-\ta)G(1+\varphi)G(1-\varphi)},
\ee
in which $G$ is the Barnes function, essentially defined by
$G(z+1)=\Gamma(z) G(z)$, $G(1)=1$. (We hope no confusion  takes place with the notation used 
in Appendix~\ref{AppendixOkamoto} for
the solution of the  Chazy $\CVI$ ODE (\ref{CVIgenrat}) which gives the metric for Bonnet surfaces.)
Note that for $\varphi=1$ (corresponding to $\xi=1$),  the Barnes function 
$G(1-\phi)$ in the denominator vanishes and,  as a signal of the formal failure of the Wiener-Hopf factorization, the amplitude above  therefore   becomes  infinite. Of course, a  finite amplitude 
exists, and 
we shall detail later how to obtain the correct, ``regularized"  result.

To compute the third term of of (\ref{BOtaphigen}), we need to 
introduce another function $f_\taphi(x)$ defined for $x>0$ through the inverse Fourier transform of the ratio of the two Wiener-Hopf factors $F^\pm(u) = F_\taphi(\pm \ii u/2)$
\be
\label{defftamux}
f_{\ta,\varphi}(x) \coloneqq \calF^{-1}\lb \frac{F^-}{F^+}-1 \rb(x) = 
\calF^{-1}\lb \frac{F^+}{F^-}-1 \rb(-x),
\ee
the second equality holding due to the evenness of the symbol.

This function $f_\taphi$ determines the  operator $\calL_{\taphi}$ appearing in the remainder term of the Borodin-Okounkov formula (\ref{BOtaphigen}) as the \textit{square of a Hankel operator}, whose  matrix elements when that $\calL_\taphi$ acts on    $L^2([\ell, +\infty))$ are 
\be 
\label{defLkernel}
\forall x>0, \forall y>0:\
\calL_{\taphi}(x,y)\!\upharpoonright_{[\ell,+\infty]}
\coloneqq \int_{\ell}^{+\infty} \!\dd w \, f_{\ta,\varphi}(x+w)f_{\ta,\varphi}(w+y).
\ee
Note in particular the  simple expression valid for arbitrary $\ell>0$ of the second derivative of the trace of this operator, 
\be
\label{D2traceL}
\frac{\dd^2}{\dd \ell^2}\Tr{\lb \calL_\taphi\rb}\!\upharpoonright_{[\ell,+\infty]} = \lp f_\taphi(\ell)\rp^2.
\ee

It remains to compute in a sufficiently explicit way 
$f_\taphi$. Using  the $\ta \leftrightarrow \varphi$ symmetry $1/F_\taphi(z) = F_{\varphi, \ta}(-z)$ of the two Wiener-Hopf factors, its
definition (\ref{defftamux}) reduces to
an integral \textit{\`a la Mellin-Barnes} involving   a ratio of well-balanced Gamma functions 
\be
f_{\ta,\varphi}(x) = 2 
\int_{-\ii \infty}^{\ii \infty}\frac{\dd z}{2 \ii \pi}
e^{-2 x z }
\lp 
\frac{\Gamma(z+\frac{1+\ta}{2})\Gamma(z+\frac{1-\ta}{2})\Gamma(-z+\frac{1+\varphi}{2})\Gamma(-z+\frac{1-\varphi}{2})}{\Gamma(z+\frac{1+\varphi}{2})\Gamma(z+\frac{1-\varphi}{2})\Gamma(-z+\frac{1+\ta}{2})\Gamma(-z+\frac{1-\ta}{2})}
-1
\rp.
\ee

Closing up the contour in the right-half plane and picking up the residues of the two strings of poles originating from the two top rightmost Gamma functions above one can express (as in \cite{BasorChen2003,BE2005}) $f_\taphi$ as the sum of two \textit{convergent} hypergeometric series ${}_4F_3$. For simplicity, we do not write the corresponding formula, since to determine the large $\ell$ behavior of the remainder term in the Borodin-Okounkov formula (\ref{BOtaphigen}), it is sufficient to keep
the first two terms in the series 
\ba
\nonumber
f_{\taphi}(\ell)= \frac{\ta^2-\varphi^2}{2} \Bigg{(} 
\frac{\Gamma\lp \frac{\ta-\varphi}{2}\rp \Gamma\lp \frac{-\ta-\varphi}{2}\rp}{\Gamma\lp \frac{\ta+\varphi}{2}\rp\Gamma\lp \frac{-\ta+\varphi}{2}\rp} 
\frac{\Gamma(\varphi)}{\Gamma(1-\varphi)}&e^{-(1-\varphi)\ell}& \\
\label{ftaphilarge}
+\frac{\Gamma\lp \frac{\ta+\varphi}{2}\rp \Gamma\lp \frac{-\ta+\varphi}{2}\rp}{\Gamma\lp \frac{\ta-\varphi}{2}\rp\Gamma\lp \frac{-\ta-\varphi}{2}\rp} 
\frac{\Gamma(-\varphi)}{\Gamma(1+\varphi)}&e^{-(1+\varphi)\ell}& +\dots \Bigg{)}.
\ea
 When $0 < \varphi <1$, the second  term  is  of course sub-dominant for $\ell \gg 1$, but we momentarily display it  in conjunction with the first one to 
 emphasize
 the overall parity symmetry $\varphi \leftrightarrow -\varphi$ expected from (\ref{defxi}), while  the omitted terms
 are all decaying smaller than $e^{- 2 \ell}$ (uniformly in $0<\varphi^2<1$).
 
 Bearing this in mind, we 
  abbreviate the  asymptotic form (\ref{ftaphilarge}) as
 \be
 \label{ftaphiabrege}
f_{\taphi}(\ell) = c_\taphi e^{-(1-\varphi)\ell} + 
\dots, \quad  c_\taphi \coloneqq \frac{\ta^2-\varphi^2}{2} \frac{\Gamma\lp \frac{\ta-\varphi}{2}\rp \Gamma\lp \frac{-\ta-\varphi}{2}\rp}{\Gamma\lp \frac{\ta+\varphi}{2}\rp\Gamma\lp \frac{-\ta+\varphi}{2}\rp} 
\frac{\Gamma(\varphi)}{\Gamma(1-\varphi)},
 \ee
 where the  sub-dominant corrections are at least
 $\calO\lp e^{-(1-\varphi)\ell}\rp$.

Working to this exponentially small lowest order  is sufficient to determine  the so-called \textit{Widom-Dyson constant}  
$\calC_\taphi$ occurring in the remainder term of the Borodin-Okounkov formula
\be
\label{BO-C-trans}
\log{\Det{\lp  \Id - \calL_{\ta,\varphi} \rp}}\!\upharpoonright_{[\ell,+\infty)} 
=  -   \calC_{\ta,\varphi} \, e^{-2(1-\varphi)\ell} + \quad
\text{terms  at  least} \, \mathcal{O}\lp e^{-2 \ell} \rp,
\ee 
with here
\be
\label{resCtaphi}
\calC_\taphi= \lp \frac{c_\taphi}{2(1-\varphi)} \rp^2 = 
\lp \frac{\varphi^2-\ta^2}{4}
\frac{\Gamma(\varphi)}{\Gamma(2-\varphi)}
\frac{\Gamma\lp \frac{\theta-\varphi}{2}\rp\Gamma\lp \frac{-\theta-\varphi}{2}\rp}
     {\Gamma\lp \frac{\theta+\varphi}{2}\rp\Gamma\lp \frac{-\theta+\varphi}{2}\rp}\rp^2.
\ee

Indeed,  the tail operator $\calL_{\taphi}\upharpoonright_{(\ell, + \infty)}$ is   ``very small'', meaning that due to (\ref{ftaphiabrege}) one can check that
\be
\Tr{\lb  \calL^n_{\ta,\varphi}\rb}\!\upharpoonright_{(\ell, + \infty)}
\, \sim \, \lp \Tr{\lb  \calL_{\ta,\varphi}\rb\!\upharpoonright_{(\ell, + \infty)}} \rp^n, \quad \ell \gg 1,
\ee
with the single  trace  itself exponentially small
\be
\label{resTrLlarge}
\mathrm{Tr}\lb 
 \calL_{\taphi}\rb\!{\upharpoonright}_{(\ell, + \infty)}= 
 \int_{0}^{+\infty} \! \dd x_{1} \int_{0}^{\infty} \!  \dd x_2  \lp f_{\taphi}(\ell+x_1+x_2)\rp^2 \, \sim \, \lp \frac{c_\taphi \, e^{-(1-\varphi)\ell}}{2(1-\varphi)} \rp^2, \quad \ell \gg 1.
\ee

This implies  that 
 the standard trace-log expansion of the  determinant 
 can be re-summed exactly in this asymptotic regime, because then
\be
\log{\Det{\lb \Id - \calL_{\ta,\varphi}\rb}}\! \upharpoonright_{(\ell, + \infty) } = - \sum_{n \ge 1} \frac{1}{n} \Tr{\lb  \calL^n_{\ta,\varphi}\rb}\!\upharpoonright_{(\ell, + \infty)} 
\, \sim \, - \sum_{n \ge 1} \frac{1}{n} \lp \Tr{\lb  \calL_{\ta,\varphi}\rb\!\upharpoonright_{(\ell, + \infty)}} \rp^n
\ee
and therefore
\be
\Det{\lp \Id -\calL_\taphi\rp}\! \upharpoonright_{[\ell,+\infty)]} \, \sim \,  e^{-\Tr{\lb \calL_\taphi\rb }\! \upharpoonright_{[\ell,+\infty)]}}=1 -
\Tr{\lb \calL_\taphi\rb }\! \upharpoonright_{[\ell,+\infty)]} + \cdots,
\ee
 up to terms decaying exponentially faster than the single trace, which by using (\ref{resTrLlarge}) thus proves (\ref{BO-C-trans})  with the value (\ref{resCtaphi}).

Putting all the pieces together,   we therefore have  in the large $\ell$ limit 
the following   expansion for the logarithm of  $\calD_\taphi(\ell)$, the three constants being  fully explicit and given  respectively by (\ref{resAtaphi}), (\ref{resBtaphi}), and (\ref{resCtaphi}) in terms of the parameters $0 \le \ta < \varphi <1$  (the limit $\ta \to 0$ never posing any problems),
\be
\label{BOtaphilarge} 
\log{\calD_\taphi(\ell)} = 
\calA_\taphi \cdot \ell 
+ \calB_\taphi 
-  \calC_{\ta,\varphi} \, e^{-2(1-\varphi)\ell} 
+  \mathcal{O}\lp e^{-2 \ell} \rp.
\ee

Under the stated conditions on the parameters, this uniformly convergent expansion is also dif\-fe\-ren\-tia\-ble with respect to $\ell$. Hence the defining Eq.~(\ref{defHtaphiDlogDet})   gives  the expression for the mean curvature function  in the large $\ell$ limit as
\be  
\label{resHtaphilarge}
\calH_\taphi(\ell) = \calA_\taphi + 2 (1-\varphi){\calC_\taphi} e^{-2(1-\varphi)\ell} + \mathcal{O}\lp e^{-2 \ell}\rp,
\ee
and thus   the asymptotic mean curvature of these Bonnet surfaces 
at their sole umbilic point:
\be
\lim_{\ell \to + \infty}\calH_\taphi(\ell) = \calA_\taphi = \frac{\ta^2-\varphi^2}{2},
\ee
a negative value approached exponentially fast from above by (\ref{resHtaphilarge}).

    \item When $\xi=1$ (i.e.~$\varphi=1$), the Wiener-Hopf symbol has a so-called \textit{Fisher-Hartwig} sin\-gu\-la\-ri\-ty \cite{FH1968}, because by (\ref{resFTKta}) it vanishes at $u=0$ with a double zero   
\be
\label{WHsing}
\widehat{W}_{\ta,1}(u)=1 - \widehat{\Kta}(u) \sim  \lp \frac{\pi}{2 \cos {(\pi \ta/2)}} \rp^2 \, u^2.
\ee
Since the logarithm of this symbol   remains integrable,  its geometric means  is simply  obtained by   continuity  from the value 
at the origin of its inverse Fourier transform $e_{\ta,1}(x)$, hence from   (\ref{resetaphix}) and (\ref{resAtaphi}) as
    \be
    a_\ta \coloneqq \lim_{\varphi \to 1} \calA_\taphi = \frac{\ta^2-1}{2}.
    \ee
    
However,  the Wiener-Hopf factorization (\ref{WHfactors}) now  breaks down,
since each of the  corresponding factors $F^\pm(u)=F_{\ta,1}(\pm \ii u/2)$ has
a zero at $u=0$,
which entails that  the BO representation formula (\ref{BOtaphigen}) does not exist anymore as it stands, as already testified by the divergence of the amplitude $\calB_\taphi$ (\ref{resBtaphi}) in the $\varphi \to 1$ limit.
    
Fortunately,  there exists a result by Widom,
who  proves   \cite[Theorem 4]{Widom2007} that for a symbol  with the simplest possible Fisher-Hartwig singularity such as (\ref{WHsing}), 
  the following {rigorous} expansion holds for its Fredholm  determinant
\be
\label{resD1strong}
\log{\calD_\ta(\ell; \xi)}\Big{|}_{\xi=1}
= a_\ta \cdot \ell + \log{\lp \ell +c_\ta \rp} +   b_\ta + 
\calO\lp e^{-(1-\ta)\ell} \rp, \quad \ell \gg 1,
\ee
with explicitly known constants $b_\ta$ and $c_\ta$,  the exponential error term being the best possible, since it comes from the existence  of an exponential moment $\int_\RR \ \dd u \ e^{p u} \widehat{K}_\ta(u)$ of order $p<1-\ta$, a technical condition  fulfilled here 
thanks to (\ref{resFTKta}). 
In particular, (\ref{resD1strong}) localizes exactly
--- i.e. with an error term  strictly zero --- for the pure Fisher-Hartwig symbol 
$u^2/(u^2+1)$
corresponding to 
the Markovian kernel $K(x-y)=e^{-|x-y|}$, and for which the truncated Fredholm determinant is equal to
$e^{-\ell}(1+\ell/2)$ \cite{Mikaelyan1986}, one of the rare cases known for all $\ell>0$.

The above work by Widom is very little cited.
For  the determinant of a Wiener-Hopf symbol with a  Fisher-Hartwig singularity such as (\ref{WHsing}) and its generalizations, one rather finds in the literature (see, e.g., \cite[Appendix A]{BK2021}, and re\-fe\-ren\-ces therein) 
a conjectural expansion under
the deprecated form 
\be
\label{resD1weak}
\log{\calD_\ta(\ell; \xi)}\Big{|}_{\xi=1}= a_\ta \cdot \ell + \log{\ell} + b_\ta + \calO \lp \ell^{-1}\rp, \quad \ell \gg 1.
\ee
This corresponds to replace formally  for $\varphi=1$  in the original  BO formula (\ref{BOtaphilarge}) the diverging
$\calB_\taphi \hookrightarrow  \log{\ell}+b_\ta $ for $\ell$ large.
Of course, (\ref{resD1weak}) is a genuine mathematical consequence of (\ref{resD1strong}) 
(and a much weaker one concerning the error term) if one expands the logarithmic piece $\log{(\ell + c_\ta)}= \log{\ell}+ \calO{(\ell^{-1})}$, 
with  here the  finite and computable constant 
\be
\frac{c_\ta}{2}= \frac{\dd }{\dd z}
\lp \frac{F_{\ta,1}(z)}{F_{\ta,1}(-z)}\rp\Big{|}_{z=0} = \Psi \lp \frac{1+\ta}{2}\rp -\Psi\lp \frac{1-\ta}{2}\rp -2 \Psi(1),
\ee
(the common zero at the origin of the two Wiener-Hopf factors simplifying in
the ratio), where $\Psi(z)=(\Gamma'/\Gamma)(z)$ is the digamma function and $\Psi(1)=-\gamma_{\mathrm{E}}$ Euler's constant.  

The regularized amplitude is  given by the now convergent  integral \footnote{We thank G.~Korchemsky for pointing \cite{BK2021} to us, and for showing us another method   to obtain (\ref{resbta}) from (\ref{resBtaphi}).}
\be
\label{resbta}
\exp{(b_\ta)} = \exp{\lp  \int_0^{\infty}\! \dd x \ x \lb 
\lp e_{\ta,1}(x) \rp ^2 
 - \frac{1-e^{-x}}{x}
\rb \rp} = 
\frac{\pi^2}{4^{\ta^2}} \frac{G^8(1/2)G(1-\ta)G(1+\ta)}{G^4(1-\ta/2)G^4(1+\ta/2)},
\ee
which is a one-parameter extrapolation of the value that
Fitzgerald, Tribe, and Zaboronsky have 
obtained in \cite{FTZ2022}  by a direct study of the Pfaffian Fredholm determinant for the   sech kernel. (The final explicit expression given in (\ref{resbta}) relies on some Barnes function duplication identities.)
In the case relevant for persistence  ($\theta=1/2$), we have found that (\ref{resbta}) evaluates in terms of the  Glaisher-Kinkelin constant $A_{\mathrm{GK}}$,   that one  also often encounters  in Painlev\'e connection problems, and which is 
related to the Riemann $\zeta$ function by $\log{A_{\mathrm{GK}}}=-\zeta'(-1)+1/12$, so that here eventually
\be
\exp{(b_\ta)}\Big{|}_{\ta=1/2} = \frac{e^{6 \zeta'(-1)}}{2^{1/12}}\, 
\frac{\Gamma(3/4)}{\Gamma(1/4)}.
\ee

Since the expansion (\ref{resD1strong}) is uniformly convergent, we 
thus obtain the asymptotic behavior
 of the Bonnet-$\PVI$  mean curvature function associated to the unthinned $\Kta$ kernel:
\be
\label{}
\calH_\ta(\ell;\xi)\Big{|}_{\xi=1} = \frac{\theta^2-1}{2} + \frac{1}{\ell} 
+  \calO\lp \ell^{-2} \rp, \quad \ell \to + \infty.
\ee

\end{itemize}

All the results when $\xi<1$ in this Appendix can be summarized in the 

\begin{Proposition}
For any $0 < \ta < \varphi  <1$, where  $\varphi=\vtaxi$ is the  uniformizing parameter defined in (\ref{defxi}), and bijectively equivalent to $\xi=\xi_\ta(\varphi)$,
the Borodin-Okounkov formula applied to the Wiener-Hopf symbol in Fourier space 
\be
{\widehat{W}}_{\taphi}(u) \coloneqq \calF\lb \Id -\xi_\ta(\varphi) \Kta \rb(u) = \frac{\cosh{(\pi u)}+\cos{(\pi \varphi)}}{\cosh{(\pi u)}+\cos{(\pi \ta)}}, \quad -(1-\ta) < \IM{u} < (1- \ta),
\ee
provides for the determinant $\calD_\taphi(\ell)$ of the truncated operator ${\widehat{W}}_{\taphi}\! \upharpoonright_\ell$
a representation valid for all interval length $\ell$
\be
\forall  \ell >0:\ 
\log{\calD_\taphi(\ell)} = \calA_\taphi \cdot \ell + \calB_\taphi + 
\log{\Det{\lb \Id - \calL_{\taphi}\rb}}\!\upharpoonright_{[\ell,+\infty)}.
\ee
\begin{itemize}
    \item The Szeg\"o-Kac-Akhiezer formula gives its 
dominant exponential behavior as $\ell  \to +\infty$
        \be
        \calA_{\taphi} \coloneqq \lim_{\ell \to + \infty} 
        \frac{\log{\calD_\taphi(\ell)}}{\ell}  
        = \frac{\theta^2 -\varphi^2}{2},
        \ee
       this (always negative) geometric means being continuous when $\varphi \nearrow 1$
        (or  $\xi \nearrow 1$ equivalently). 
        \item The  subleading term 
        $\calB_{\taphi}$  in (\ref{BOtaphigen})    determines the
        asymptotic   amplitude
        of the determinant, and this prefactor
        evaluates in terms of  Barnes' $G$-function to
        \be
         \exp{\calB_{\taphi}} =\frac{  G^2\lp 1+\frac{\ta+\varphi}{2}\rp G^2\lp 1-\frac{\ta+\varphi}{2}\rp G^2\lp 1+\frac{\ta-\varphi}{2}\rp G^2\lp 1-\frac{\ta-\varphi}{2}\rp}{G(1+\ta)G(1-\ta)G(1+\varphi)G(1-\varphi)}.
        \ee
        \item 
         The remainder term involves the  square of a bounded Hankel operator $\calL_{\taphi}$  acting on the complementary interval $[\ell,+\infty)$, whose   complicated but explicit expression is determined by the Wiener-Hopf factors of the symbol $W_\taphi$. Its determinant decays exponentially fast
        \be
        \log{\Det{\lp  \Id - \calL_{\ta,\varphi} \rp}}\!\upharpoonright_{[\ell,+\infty)} =  -   \calC_{\ta,\varphi}e^{-2(1-\varphi)\ell} + \mathcal{O}\lp e^{-2 \ell}\rp,
        \ee
           where $\calC_{\ta,\varphi} >0$ is the associated Widom-Dyson constant:
          \be
\calC_{\ta,\varphi}= \lp \frac{\ta^2-\varphi^2}{4(\varphi-1)}\frac{\Gamma(\varphi)}{\Gamma(1-\varphi)}\frac{\Gamma\lp \frac{\theta-\varphi}{2}\rp\Gamma\lp \frac{-\theta-\varphi}{2}\rp}{ \Gamma\lp \frac{\theta+\varphi}{2}\rp\Gamma\lp \frac{-\theta+\varphi}{2}\rp}\rp^2.
          \ee

\end{itemize}

\end{Proposition}

\end{document}